%% file: Quantum_Coarse-Graining_Symmetries_and_Reducibility_of_Dynamics_PRA_v2.tex
\documentclass[10pt,twocolumn,english,nofootinbib]{revtex4-1}
\usepackage{lmodern}
\usepackage[T1]{fontenc}
\usepackage[latin9]{inputenc}
\usepackage{color}
\usepackage{babel}
\usepackage{float}
\usepackage{amsmath}
\usepackage{amsthm}
\usepackage{amssymb}
\usepackage{stmaryrd}
\usepackage{graphicx}
\usepackage[unicode=true,pdfusetitle,
 bookmarks=true,bookmarksnumbered=true,bookmarksopen=false,
 breaklinks=false,pdfborder={0 0 0},pdfborderstyle={},backref=false,colorlinks=true]
 {hyperref}
\usepackage{breakurl}

\makeatletter

\providecommand{\tabularnewline}{\\}

  \theoremstyle{plain}
  \newtheorem{thm}{\protect\theoremname}
 \ifx\proof\undefined\
   \newenvironment{proof}[1][\proofname]{\par
     \normalfont\topsep6\p@\@plus6\p@\relax
     \trivlist
     \itemindent\parindent
     \item[\hskip\labelsep
           \scshape
       #1]\ignorespaces
   }{%
     \endtrivlist\@endpefalse
   }
   \providecommand{\proofname}{Proof}
 \fi
  \theoremstyle{plain}
  \newtheorem{cor}{\protect\corollaryname}
  \theoremstyle{plain}
  \newtheorem{prop}{\protect\propositionname}
  \theoremstyle{plain}
  \newtheorem{lem}{\protect\lemmaname}

\renewcommand{\fnum@figure}{FIG.~\thefigure}
\definecolor{green}{RGB}{0,225,0}

\AtBeginDocument{
  
}

\makeatother

\providecommand{\corollaryname}{Corollary}
\providecommand{\lemmaname}{Lemma}
\providecommand{\propositionname}{Proposition}
\providecommand{\theoremname}{Theorem}

\begin{document}
\input{MathMacros.tex}

\title{Quantum Coarse-Graining, Symmetries and Reducibility of Dynamics}

\author{Oleg Kabernik}
\email{olegk@phas.ubc.ca}

\affiliation{Department of Physics and Astronomy, University of British Columbia,
Vancouver, BC, Canada}

\date{\today}
\begin{abstract}
The common idea behind complexity reduction in physical systems is
separation of information into ``physically meaningful'' and ``safely
ignorable''. Here we consider a generic notion of such separation
\textendash{} implemented by coarse-graining the state-space \textendash{}
and address the question of what information is indeed safely ignorable
if we want to reduce the complexity of dynamics. The general condition
for reducibility of dynamics under coarse-graining will be presented
for stochastic and quantum systems. In the process we develop the
quantum notion of state-space coarse-graining that allows to marginalize
selected degrees of freedom. One of our main findings is that there
is a broader class of symmetries, beyond those that are considered
by Noether's Theorem, that can play a role in the reduction of dynamics.
Some examples of quantum coarse-grainings and the reduction of dynamics
with symmetries will be discussed. 
\end{abstract}
\maketitle

\section{Introduction}

The complexity presented by real physical systems is a fundamental
challenge that often resists ``brute force'' calculations but is
occasionally manageable with some analytical insight. The idea of
coarse-graining (CG) is a prime example of such insight, and its use
in physics traces back to the Ehrenfest's work on statistical mechanics
\cite{Ehrenfest}. Today, there are many forms in which CG appears
in physics: renormalization methods in condensed matter \cite{Kadanoff66,Wilson75},
coarse-grained modeling of biomolecular dynamics \cite{Ing=0000F3lfsson14},
and separation of scales in cosmology \cite{Ellis05} are some of
the common examples. Nonetheless, there is a common, system independent,
notion of CG that underlines all such approaches, and that is the
abstract notion of \textit{state-space coarse-graining} (from here
on by ``CG'' we will refer to this abstract notion). Studying the
implications of such generic notion of CG is therefore essential for
our understanding of complexity reduction in physical systems on a
fundamental level. 

The notion of CG is an elementary proposition in statistical mechanics
which asserts that if one is unable to distinguish some states of
the system, then the system is described by a smaller (coarser) state-space
of distinguishable states. In the context of thermodynamics, CG is
manifested by our inability to measure micro states of the system,
leading to the definition of macro states described by variables such
as temperature and pressure. Another common manifestation of CG is
the situation where a composite system has an inaccessible subsystem.
Our inability to distinguish between states that differ only by the
inaccessible part leads to a coarser description which we account
for by marginalizing the inaccessible subsystem. 

Despite its origin as a manifestation of practical limitations, the
notion of CG is generic, specified only by the choice of indistinguishable
states. Therefore, we can consider CG as a generic way to introduce
ignorance without relying on any physical structure of the system. 

The simple classical notion of CG does not translate naturally into
quantum theory and recently there have been multiple proposals for
its extension. In \cite{Matteo17}, quantum CG was implemented by
coarse-graining the quasi-probability (Wigner function) representation
of an $N$ qubit system. The authors of \cite{Duarte17} argue that
any dimension reducing quantum channel can be interpreted as a quantum
CG. In \cite{Singh17}, CG of the Hilbert space was specified by a
set of preferred states and implemented with the statistical method
of Principal Component Analysis. Finally, in \cite{Faist16}, the
quantum notion of CG was presented as the effective state-space perceived
by a constrained observer. 

The goal of current work is twofold: (a) establish the quantum notion
of CG by direct analogy with the classical concept and provide it
with operational meaning; (b) develop the framework for complexity
reduction of dynamics with CG and integrate it with the framework
of symmetries. We will initially work out the main concepts in the
more intuitive setting of classical stochastic systems, and then proceed
to the finite dimensional quantum setting. The stochastic case will
be accompanied by analysis of out-of-equilibrium dynamics of a 1D
Ising chain. In the quantum setting we will discuss some special cases
of the CG map and analyze the dynamics of continuous time quantum
walk on a binary tree. 

In order to formulate the quantum notion of CG as closely as possible
to the classical case, we will first establish it in the context of
stochastic systems. The key observation here is that CG can be interpreted
as a marginalization of a kind of subsystem (we will call it \textit{partial}
subsystem and it generalizes the idea of \textit{virtual} subsystem
\cite{Zanardi01}). This will allow us to formulate the quantum notion
of CG by a direct analogy. The result is a dimension reducing map
that implements a quantum CG scheme according to specifications that
resemble the classical choice of indistinguishability. Furthermore,
the specification of quantum CG will be directly related to a restricted
set of observables that give it operational meaning. 

The main application we will focus on is reduction of dynamics. The
key problem is identifying such CGs that allow time evolutions in
the reduced state-space to be governed by a reduced generator of dynamics.
This can be summarized with the diagram in fig. \ref{fig:dynamics reduction diagram}.
In general, time evolutions in the reduced state-space are not even
uniquely determined by initial conditions, and when they are they
may still loose the semigroup structure necessary for the existence
of the generator of dynamics \cite{Rivas12}. Therefore, it is important
to understand the compatibility condition between CG and dynamics
that allows the preservation of semigroup structure in the reduced
state-space. We will provide the general version of such condition
in Theorem \ref{thm:PQ=00003DPQP}, which applies to both stochastic
and quantum systems, and specialize it to unitary dynamics in Theorem
\ref{thm:=00005BH,S=00005D in span=00007BS_kl=00007D}. 

Symmetries turn out to play an important role in the analysis of reducibility
of dynamics. We will see that symmetrization of the state-space with
respect to some group representation is a special case of CG. Inserting
this case into the general compatibility conditions between CG and
dynamics leads to a broader class of symmetries relevant in the analysis
of dynamical evolutions. The new symmetries are defined in theorems
\ref{thm:SUM(=00005BD(g),Q=00005D)=00003D0 <=00003D> PQ=00003DPQP}
(stochastic) and \ref{thm:comp. of symmetrization with dynamocs}
(quantum) by a compatibility condition with the generator of dynamics.
In both stochastic and quantum cases, the compatibility condition
extends the relevant symmetries beyond those that commute with time
evolutions, as considered by Noether's Theorem.

\begin{figure}[t]
\centering{}\includegraphics[width=1\columnwidth]{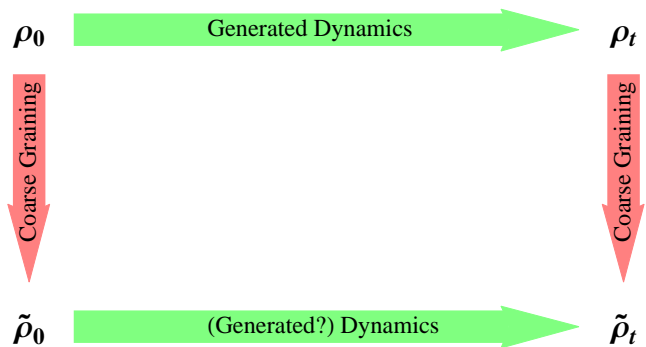}
\caption{\label{fig:dynamics reduction diagram}Dynamics of coarse-grained
states cannot be generated unless the coarse-graining scheme is \textit{compatible
}with the original generator of dynamics. }
\end{figure}

\section{Classical Coarse-Graining}

Before we formally define CG for discrete stochastic systems, it is
worth having a concrete, albeit generic, example. 

Consider a random walk on the graph of fig. \ref{fig:random walk-1}(a).
The weights on the edges represent the rate (probability per unit
time) of transitions between connected vertices in both directions.
Parameters $a$, $b$, $c$, $d$, $e$ are all positive and the rest
obey $\left|\delta\right|,\left|\epsilon\right|\leq\frac{c}{2}$,
s.t. all rates are non-negative. If we coarse-grain this system by
choosing not to distinguish between vertices that appear in the same
column, then we partition it into 3 blocks associated with the macro
states
\[
u_{1}=\left\{ v_{1}\right\} \,\,\,\,u_{2}=\left\{ v_{2},v_{3}\right\} \,\,\,\,u_{3}=\left\{ v_{4},v_{5},v_{6}\right\} .
\]
The question now is what values can we assign, if any, to the transition
rates between the macro states.

If we consider a single vertex, say $v_{2}$, and sum all the transition
rates from $v_{2}$ to the column on the right, we get 
\[
\sum_{v\in u_{3}}r\left(v_{2}\rightarrow v\right)=3c.
\]
We will get the same value if instead of $v_{2}$ we take $v_{3}$.
Therefore, the rate of transitions from \textit{any} vertex in the
middle column, to the right column is $3c$. This unambiguously defines
the rate of transition from the middle column to the right column,
without reference to any particular vertex. Similarly, the rate of
transitions from any vertex in the right column to the middle column
is $2c$. We can repeat this argument for transitions between the
left and the middle columns, yielding the rates of $2a$ and $a$
in the opposing directions. We should also note that there are no
direct transitions between the left-most and the right-most columns.
Therefore, transition rates between all three columns are well defined
and shown in fig. \ref{fig:random walk-1}(b).

The fact that we can get a well defined random walk in the reduced
state-space is not trivial. Such reduction of dynamics is only possible
when the rate of transitions between the chosen macro states is unambiguous.
Choosing a slightly different CG, where the macro state are
\[
u_{1}=\left\{ v_{1},v_{2},v_{3}\right\} \,\,\,\,\,\,\,\,\,\,\,\,u_{2}=\left\{ v_{4},v_{5},v_{6}\right\} ,
\]
results in undefined transition rates. That is because the rate of
transitions from $v_{1}$ to any vertex in $u_{2}$ is $0$, but from
$v_{2}$ or $v_{3}$ it is $3c$. Given that the initial macro state
is $u_{1}$, it is impossible to tell what the initial rate of transitions
to $u_{2}$ will be, because it depends on where inside $u_{1}$ it
actually starts. Similarly, the original choice of CG by the columns
would not work if we slightly change the dynamics by altering the
transition rate between $v_{1}$ and $v_{2}$ to $\tilde{a}\neq a$.
Now, it is not possible to tell the rate of transitions from the middle
column to the left because it depends on the internal state of the
column. 

This example demonstrates the fact that it may be possible to generate
time evolutions in the coarse-grained state-space but the original
dynamics and the CG have to be compatible. Such compatibility does
not imply that the rates of transition must be uniform \textendash{}
in general all 6 rates between vertices in the middle and the right
columns in fig. \ref{fig:random walk-1}(a) can be different. We will
prove that the necessary and sufficient condition for such compatibility
in Markovian stochastic systems is what we already noted: the \textit{total}
rate of transitions from \textit{any} state in one block to another
should be constant. 
\begin{figure}[t]
\begin{raggedright}
(a) %
\noindent\begin{minipage}[t]{1\columnwidth}%
\includegraphics[width=1\columnwidth]{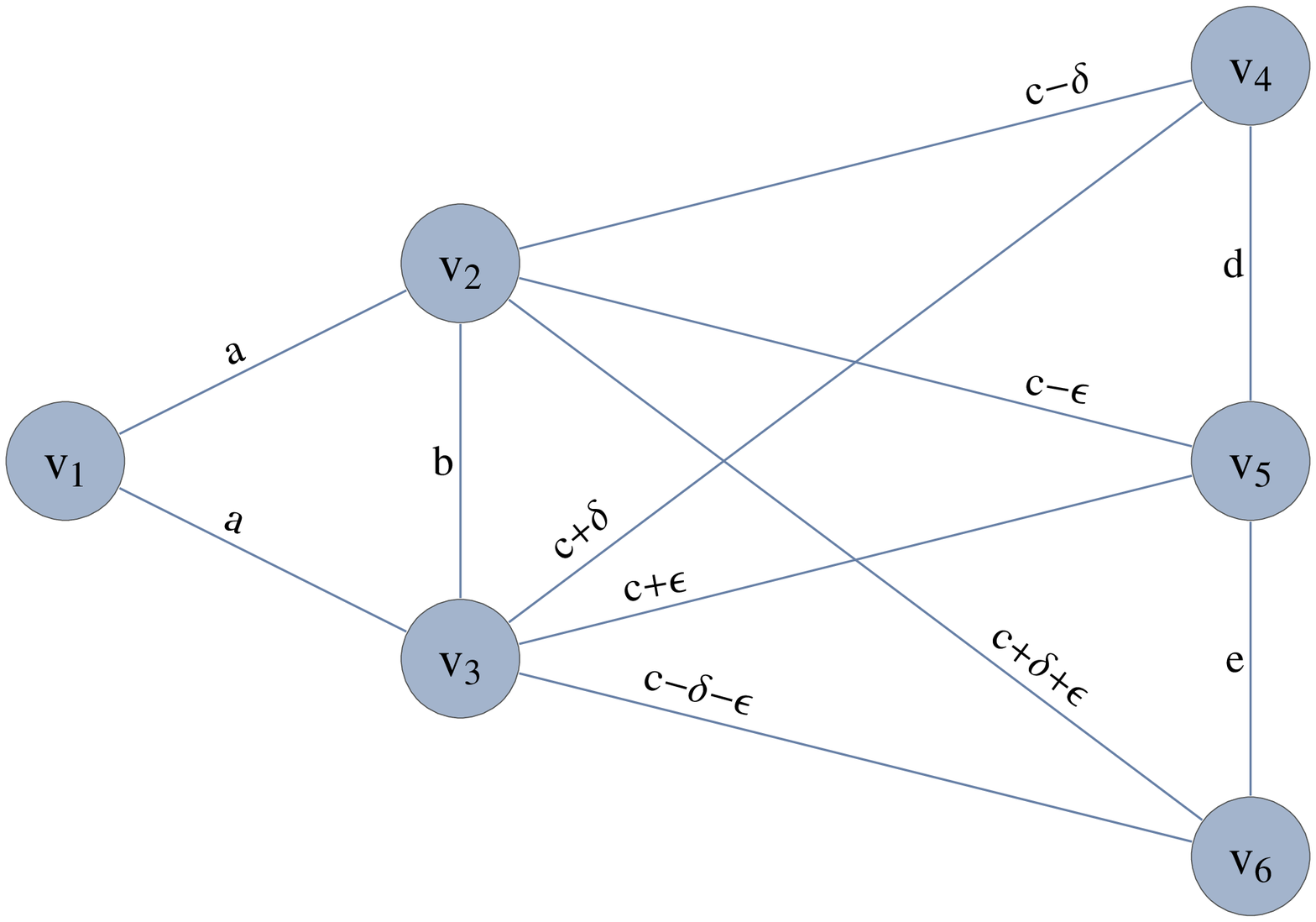}%
\end{minipage}
\par\end{raggedright}
\raggedright{}(b) %
\noindent\begin{minipage}[t]{1\columnwidth}%
\includegraphics[width=1\columnwidth]{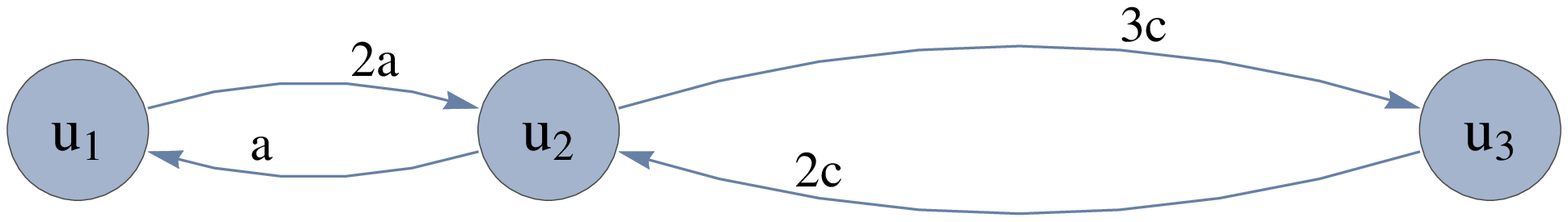}%
\end{minipage}\caption{\label{fig:random walk-1}(a) Random walk graph with transition rates
specified on the edges. (b) Coarse-grained graph where the vertices
in the same column are blocked together. Effective transition rates
between columns are specified on directed edges. }
\end{figure}

\subsection{Formal Definition\label{subsec:Basic-Definitions}}

Classical indistinguishability of states can be specified by an equivalence
relation that partitions the state-space into equivalence classes
of \textit{macro} states. If we consider the system with a discrete
and finite state-space $\mathcal{\mathsf{A}}:=\left\{ \alpha_{i}\right\} _{i=1}^{\mathsf{\left|A\right|}}$,
we can specify its CG by the set $\mathcal{\mathsf{B}}:=\left\{ \beta_{k}\right\} _{k=1}^{\mathsf{\left|B\right|}}$
of disjoint subsets of $\mathcal{\mathsf{A}}$ that partition $\mathcal{\mathsf{A}}$.
If we order the set $\mathcal{\mathsf{A}}$ consistently with the
partition we can identify blocks of indices $b_{k}:=\left\{ i_{k},i_{k}+1,...\right\} $
such that $\beta_{k}:=\left\{ \alpha_{i}\right\} _{i\in b_{k}}$.
The system is said to be in a macro state $\beta_{k}$ if it is in
\textit{any} of the micro states $\alpha_{i}\in\beta_{k}$. 

Probabilistic micro (macro) states of the system live in the vector
space $\mathbb{R}^{\mathcal{\mathsf{A}}}$ ($\mathbb{R}^{\mathcal{\mathsf{B}}}$)
of real valued functions $p$ from $\mathcal{\mathsf{A}}$ ($\mathcal{\mathsf{B}}$)
to $\mathbb{R}$ which, if positive and normalized, are interpreted
as probability distributions over the states $\mathcal{\mathsf{A}}$
($\mathcal{\mathsf{B}}$). By definition of the macro states, the
probability of finding the system in a macro state $\beta$, is the
probability of finding it in \textit{any} micro state $\alpha\in\beta$,
that is
\begin{equation}
p_{\mathcal{\mathsf{B}}}\left(\beta\right)=\sum_{\alpha\in\beta}p_{\mathcal{\mathsf{A}}}\left(\alpha\right).\label{eq: p_B(b)=00003DSum_a_(p_S(a))}
\end{equation}
Since $\mathbb{R}^{\mathcal{\mathsf{A}}}$ ($\mathbb{R}^{\mathcal{\mathsf{B}}}$)
is just an $\mathsf{\left|A\right|}$ ($\mathsf{\left|B\right|}$)
dimensional vector space, we can express relation (\ref{eq: p_B(b)=00003DSum_a_(p_S(a))})
as a vector equation $p_{\mathcal{\mathsf{B}}}=Mp_{\mathcal{\mathsf{A}}}$
and $M$ is a $\mathsf{\left|B\right|}\times\mathsf{\left|A\right|}$
block diagonal matrix of the form
\begin{equation}
M:=\begin{pmatrix}1 & \cdots & 1\\
 &  &  & \ddots\\
 &  &  &  & 1 & \cdots & 1
\end{pmatrix},\label{eq:def of M}
\end{equation}
where the $k$th block is a $1\times\left|\beta_{k}\right|$ row-vector
filled with $1$'s. $M$ acts by summing the fine-grained probability
distribution in each block of micro states into a single value, which
is the total probability of finding the system in \textit{any} micro
state of that block. We will call such $M$ a \textit{coarse-graining
matrix.}

Any CG matrix $M$ admits the right inverse $M^{+}$ such that $MM^{+}=I$
is an identity on $\mathbb{R}^{\mathcal{\mathsf{B}}}$. It is easy
to check that it is the $\mathsf{\left|A\right|}\times\mathsf{\left|B\right|}$
block diagonal matrix of the form
\begin{equation}
M^{+}:=\begin{pmatrix}\left|\beta_{1}\right|^{-1}\\
\vdots\\
\left|\beta_{1}\right|^{-1}\\
 & \ddots\\
 &  &  & \left|\beta_{\mathsf{\left|B\right|}}\right|^{-1}\\
 &  &  & \vdots\\
 &  &  & \left|\beta_{\mathsf{\left|B\right|}}\right|^{-1}
\end{pmatrix}\label{eq:def of M+}
\end{equation}
where the $k$th block is a $\left|\beta_{k}\right|\times1$ column-vector
filled with $\left|\beta_{k}\right|^{-1}$. This is the Moore-Penrose
pseudo inverse \cite{Penrose54} of $M$, which means that $P:=M^{+}M$
is an orthogonal projection on the subspace $\left(\mathsf{ker}M\right)^{\perp}\subseteq\mathbb{R}^{\mathcal{\mathsf{A}}}$.
Moreover, restriction of $M$ to $\left(\mathsf{ker}M\right)^{\perp}$
is an isomorphism $M:\left(\mathsf{ker}M\right)^{\perp}\longrightarrow\textrm{\ensuremath{\mathsf{im}}}\,M$
and since $\left(\mathsf{ker}M\right)^{\perp}=\textrm{\ensuremath{\mathsf{im}}}\,P$
and $\textrm{\ensuremath{\mathsf{im}}}\,M=\mathbb{R}^{\mathcal{\mathsf{B}}}$
it follows that $\textrm{\ensuremath{\mathsf{im}}}\,P\cong\mathbb{R}^{\mathcal{\mathsf{B}}}$.
The isomorphism between the image of $P$ and the image of $M$ implies
that $P$ erases the same fine grained information as $M$. We will
call $P$ a \textit{coarse-graining projection} which has the block
diagonal form 
\begin{equation}
P:=\begin{pmatrix}\left|\beta_{1}\right|^{-1} & \ldots & \left|\beta_{1}\right|^{-1}\\
\vdots & \ddots & \vdots\\
\left|\beta_{1}\right|^{-1} & \ldots & \left|\beta_{1}\right|^{-1}\\
 &  &  & \ddots\\
 &  &  &  & \left|\beta_{\mathsf{\left|B\right|}}\right|^{-1} & \ldots & \left|\beta_{\mathsf{\left|B\right|}}\right|^{-1}\\
 &  &  &  & \vdots & \ddots & \vdots\\
 &  &  &  & \left|\beta_{\mathsf{\left|B\right|}}\right|^{-1} & \ldots & \left|\beta_{\mathsf{\left|B\right|}}\right|^{-1}
\end{pmatrix},\label{eq:def of P}
\end{equation}

\noindent where the $k$th block is a $\left|\beta_{k}\right|\times\left|\beta_{k}\right|$
matrix filled with $\left|\beta_{k}\right|^{-1}$. $P$ acts by averaging
over the probabilities in each block. The advantage of $P$ over $M$
is that $P$ leaves the result of CG in the subspace of $\mathbb{R}^{\mathcal{\mathsf{A}}}$,
which allows a direct comparison of states before and after CG.

\subsection{Compatibility with Dynamics\label{subsec:classical Compatibility-with-Dynamics}}

In this section we show how this notion of CG allows to study the
dynamics of some select properties of the system without the need
to understand the dynamics of all its degrees of freedom. We will
focus on continuous time Markov processes (CTMP) because they are
common in classical models and are closely related to quantum dynamics.

The idea is that given a dynamical system we can coarse-grain it and
derive new dynamical rules that generate consistent time evolutions
in the coarse-grained state-space (see fig. \ref{fig:dynamics reduction diagram}).
We will say that CG is \textit{compatible }with dynamics if there
is a generator that governs time evolutions in the reduced state-space.
The main question that we address here is how to recognize compatible
CGs and how to derive the reduced generator. 

The dynamical rules of CTMP can be specified with a \textit{transition
rate matrix} $Q$ such that \cite{Breuer02}
\begin{equation}
\frac{d}{dt}p_{\mathcal{\mathsf{A}}}=Q\,p_{\mathcal{\mathsf{A}}}\label{eq:Markovian dynamics diff-eq for a classical state}
\end{equation}
for $p_{\mathcal{\mathsf{A}}}\in\mathbb{R}^{\mathcal{\mathsf{A}}}$.
The off-diagonal elements $Q_{ij}$ specify the rate of transitions
between states $\alpha_{j}\longmapsto\alpha_{i}$ while the diagonal
elements $Q_{jj}:=-\sum_{i\neq j}Q_{ij}$ specify the total rate of
transitions out of states $\alpha_{j}$. For an initial probabilistic
state $p_{\mathcal{\mathsf{A}}}\left(0\right)$, the subsequent states
are given by the solutions of Eq. (\ref{eq:Markovian dynamics diff-eq for a classical state})
as $p_{\mathcal{\mathsf{A}}}\left(t\right)=e^{tQ}p_{\mathcal{\mathsf{A}}}\left(0\right)$,
where $Q$ generates time evolutions similarly to the Hamiltonian
in quantum mechanics (strictly speaking $Q$ is closer in nature to
the Lindblad operator rather than the Hamiltonian). 

Now, consider a CG $\mathcal{\mathsf{B}}$ of $\mathcal{\mathsf{A}}$
represented by the matrix $M:\mathbb{R}^{\mathcal{\mathsf{A}}}\longrightarrow\mathbb{R}^{\mathcal{\mathsf{B}}}$.
The coarse-grained probabilistic state evolves according to $p_{\mathcal{\mathsf{B}}}\left(t\right):=Mp_{\mathcal{\mathsf{A}}}\left(t\right)$
and a priori there is no reason to assume that it also evolves as
a CTMP. However, this is exactly what we require from the compatibility
of CG with dynamics in order to be able to generate time evolutions
in the reduced state-space. The following theorem provides the necessary
and sufficient conditions for it to be true.
\begin{thm}
\label{thm:PQ=00003DPQP}Let $Q$ be a transition rate matrix as in
Eq. (\ref{eq:Markovian dynamics diff-eq for a classical state}),
let $M$ be a coarse-graining matrix as in Eq. (\ref{eq:def of M})
and let $P$ be a coarse-graining projection as in Eq. (\ref{eq:def of P}).
Then, the coarse-grained state $p_{\mathcal{\mathsf{B}}}:=Mp_{\mathcal{\mathsf{A}}}$
evolves as a CTMP for all $p_{\mathcal{\mathsf{A}}}$ if and only
if 
\begin{equation}
PQ=PQP.\label{eq:PQ=00003DPQP}
\end{equation}
The reduced transition rate matrix $\tilde{Q}$ such that $\frac{d}{dt}p_{\mathcal{\mathsf{B}}}=\tilde{Q}\,p_{\mathcal{\mathsf{B}}}$
is then given by $\tilde{Q}:=MQM^{+}$.
\end{thm}
\begin{proof}
If $PQ=PQP$ then multiplying it by $M$ on the left we get $MQ=MQM^{+}M$
and therefore 
\[
\frac{d}{dt}p_{\mathcal{\mathsf{B}}}=M\frac{d}{dt}p_{\mathcal{\mathsf{A}}}=MQp_{\mathcal{\mathsf{A}}}=\tilde{Q}p_{\mathcal{\mathsf{B}}}
\]
where \textit{$\tilde{Q}=MQM^{+}$}. This proves the ``if'' direction.

On the other hand if $p_{\mathcal{\mathsf{B}}}$ evolves as a CTMP
then there is a $\tilde{Q}$ such that $\frac{d}{dt}p_{\mathcal{\mathsf{B}}}=\tilde{Q}\,p_{\mathcal{\mathsf{B}}}$
. Therefore
\[
MQp_{\mathcal{\mathsf{A}}}=M\frac{d}{dt}p_{\mathcal{\mathsf{A}}}=\frac{d}{dt}p_{\mathcal{\mathsf{B}}}=\tilde{Q}\,p_{\mathcal{\mathsf{B}}}=\tilde{Q}Mp_{\mathcal{\mathsf{A}}}.
\]
Since it has to hold for all $p_{\mathsf{A}}$ we are left with $MQ=\tilde{Q}M$.
Multiplying it by $M^{+}$ from the right we get $\tilde{Q}=MQM^{+}$.
If we substitute $\tilde{Q}$ back into $MQ=\tilde{Q}M$ and multiply
by $M^{+}$ from the left we get $PQ=PQP$. Hence the ``only if''.
\end{proof}
Thus, for example, in the case of random walk of fig. \ref{fig:random walk-1}(a)
with CG by the columns, we have the CG matrix 
\[
M:=\begin{pmatrix}1\\
 & 1 & 1\\
 &  &  & 1 & 1 & 1
\end{pmatrix}\,\,\,\,\,\,\,\,\,\,\,\,\,\,\,\,\,\,\,\,\,\,\,M^{+}:=\begin{pmatrix}1\\
 & 1/2\\
 & 1/2\\
 &  & 1/3\\
 &  & 1/3\\
 &  & 1/3
\end{pmatrix}.
\]
And the transition rate matrix (diagonal elements are just the negatives
of the column's sum)

\[
Q:=\begin{pmatrix}Q_{11} & a & a & 0 & 0 & 0\\
a & Q_{22} & b & c-\delta & c-\epsilon & c+\delta+\epsilon\\
a & b & Q_{33} & c+\delta & c+\epsilon & c-\delta-\epsilon\\
0 & c-\delta & c+\delta & Q_{44} & d & 0\\
0 & c-\epsilon & c+\epsilon & d & Q_{55} & e\\
0 & c+\delta+\epsilon & c-\delta-\epsilon & 0 & e & Q_{66}
\end{pmatrix}.
\]
It is straight forward to check that the compatibility condition (\ref{eq:PQ=00003DPQP})
holds and the reduced transition rate matrix is
\[
\tilde{Q}=MQM^{+}=\begin{pmatrix}-2a & a & 0\\
2a & -3c-a & 2c\\
0 & 3c & -2c
\end{pmatrix}.
\]
This matrix generates the random walk of fig. \ref{fig:random walk-1}(b).

The compatibility condition (\ref{eq:PQ=00003DPQP}) has an intuitive
interpretation. If we understand the image of $P$ to be the subspace
of coarse-grained states, then the image of $P^{\perp}:=I-P$ must
be the subspace containing fine-grained information. The compatibility
condition $PQ=PQP$ is equivalent to $PQP^{\perp}=0$ which means
that $Q$ does not map fine grained information into the coarse-grained
subspace. Then, time evolution of the coarse-grained state cannot
be affected by the fine grained details. This is a sensible requirement
because if fine grained details could affect coarse-grained evolution,
it would not be possible to throw them away and expect to tell how
the coarse-grained state will evolve. 

In more concrete terms, what the compatibility condition ensures is
that the rate of transitions between macro states is a well defined
property. To see that, let's try to derive the rate of transition
between macro states from the original rates between micro states.
The total rate of transitions from a micro state $\alpha_{i}\in\beta_{k}$
to any micro state in $\beta_{k'}$ ($k\neq k'$) is given by $r\left(\alpha_{i}\longmapsto\beta_{k'}\right):=\sum_{j\in b_{k'}}Q_{ji}$.
If the value of $r\left(\alpha_{i}\longmapsto\beta_{k'}\right)$ varies
with different $\alpha_{i}\in\beta_{k}$ then knowledge of the initial
macro state $\beta_{k}$ is not enough to tell the rate of transitions
to $\beta_{k'}$. But if $r\left(\alpha_{i}\longmapsto\beta_{k'}\right)$
is the same for all $\alpha_{i}\in\beta_{k}$ then it does not matter
in which micro state of $\beta_{k}$ we start, the rate of transitions
from $\beta_{k}$ to $\beta_{k'}$ is given by $r\left(\alpha_{i}\longmapsto\beta_{k'}\right)$
for any $\alpha_{i}\in\beta_{k}$. Therefore, the notion of rate of
transitions between macro states is meaningless unless the rates $r\left(\alpha_{i}\longmapsto\beta_{k'}\right)$
are uniform over $\alpha_{i}\in\beta_{k}$ for all $\beta_{k}$ and
$\beta_{k'}$. The following corollary to Theorem \ref{thm:PQ=00003DPQP}
formalizes this argument.
\begin{cor}
\label{cor:total rates must be uniform}Let $Q$ be a transition rate
matrix and
\[
r\left(\alpha_{i}\longmapsto\beta_{k'}\right):=\sum_{j\in b_{k'}}Q_{ji}
\]
the total rate of transitions from $\alpha_{i}\in\beta_{k}$ to\textup{
$\beta_{k'}$.} Then the compatibility condition $PQ=PQP$ is equivalent
to $r\left(\alpha_{i}\longmapsto\beta_{k'}\right)$ being constant
over the subset $\beta_{k}$, for all $\beta_{k}$ and $\beta_{k'}$. 
\end{cor}
\begin{proof}
By definition of $P$ (Eq. (\ref{eq:def of P})) we calculate the
matrix elements of $PQ$ to be \texttt{
\[
\left(PQ\right)_{i'i}=\sum_{j=1}^{\mathsf{\left|A\right|}}P_{i'j}Q_{ji}=\frac{1}{\left|\beta_{k'}\right|}\sum_{j\in b_{k'}}Q_{ji}=\frac{r\left(\alpha_{i}\longmapsto\beta_{k'}\right)}{\left|\beta_{k'}\right|}
\]
}where $k'$ is the index of the block s.t. $i'\in b_{k'}$. Similarly
the matrix elements of $PQP$ are\texttt{
\begin{align*}
\left(PQP\right)_{i'i} & =\sum_{j=1}^{\mathsf{\left|A\right|}}\left(PQ\right)_{i'j}P_{ji}=\frac{1}{\left|\beta_{k}\right|}\sum_{j\in b_{k}}\left(PQ\right)_{i'j}\\
 & =\frac{1}{\left|\beta_{k}\right|\left|\beta_{k'}\right|}\sum_{j\in b_{k}}r\left(\alpha_{j}\longmapsto\beta_{k'}\right)
\end{align*}
}where $k$ is the index of the block s.t. $i\in b_{k}$. Element-wise,
the condition $PQ=PQP$ then reads 
\[
r\left(\alpha_{i}\longmapsto\beta_{k'}\right)=\frac{1}{\left|\beta_{k}\right|}\sum_{j\in b_{k}}r\left(\alpha_{j}\longmapsto\beta_{k'}\right).
\]
The right hand side depends on $i$ only through the block index $k$,
therefore this condition states that $r\left(\alpha_{i}\longmapsto\beta_{k'}\right)$
is constant for all $\alpha_{i}\in\beta_{k}$. 

On the other hand if $r\left(\alpha_{i}\longmapsto\beta_{k'}\right)$
is constant for all $\alpha_{i}\in\beta_{k}$, then 
\begin{align*}
\left(PQP\right)_{i'i} & =\frac{1}{\left|\beta_{k}\right|\left|\beta_{k'}\right|}\sum_{j\in b_{k}}r\left(\alpha_{j}\longmapsto\beta_{k'}\right)\\
 & =\frac{1}{\left|\beta_{k'}\right|}r\left(\alpha_{i}\longmapsto\beta_{k'}\right)=\left(PQ\right)_{i'i}.
\end{align*}
\end{proof}
It is worth pointing out the compatibility condition of Corollary
\ref{cor:total rates must be uniform} is the defining property of
an \textit{equitable partition }of a weighted graph specified by $Q$
\cite{Barrett17}. The problem of finding a CG compatible with dynamics
is therefore equivalent to the problem of finding an equitable partition
of the weighted graph specified by $Q$. 

\subsection{Coarse-Graining and Symmetries\label{subsec:Associated-Symmetries}}

So far we have specified CG with the choice of equivalence classes
that determine indistinguishable states. In the following we show
how CG can also be specified with group representations. 
\begin{prop}
\label{prop: P=00003DP_G}Let $G$ be a group that permutes elements
of the state-space $\mathcal{\mathsf{A}}$, and let the permutation
matrices $D\left(G\right)$ be its representation on $\mathbb{R}^{\mathcal{\mathsf{A}}}$.
Then, the symmetrizer $P:=\left|G\right|^{-1}\sum_{g\in G}D\left(g\right)$
is a coarse-graining projection associated with partition of $\mathcal{\mathsf{A}}$
into orbits of $G$. 
\end{prop}
\begin{proof}
To show that $P$ is a CG projection for any $G$ it is sufficient
to show that it acts as a CG projection on any basis element $\hat{\alpha}\in\mathbb{R}^{\mathcal{\mathsf{A}}}$
(for clarity we omit the element's index). For a given $\hat{\alpha}$
we define the subgroup that stabilizes it as $G_{\alpha}:=\left\{ g\in G\,|\,g\left(\alpha\right)=\alpha\right\} $.
Since cosets of $G_{\alpha}$ form a partition of $G$ we can write
\[
P=\left|G\right|^{-1}\sum_{C\in G/G_{a}}\,\,\,\sum_{g\in C}D\left(g\right)
\]
where $C$ runs over all distinct cosets. If $g_{1},g_{2}\in C$ belong
to the same coset of $G_{\alpha}$ then clearly $g_{1}\left(\alpha\right)=g_{2}\left(\alpha\right)$.
On the other hand if $g_{1}$, $g_{2}$ belong to different cosets
then $g_{1}\left(\alpha\right)=g_{2}\left(\alpha\right)$ implies
$g_{1}^{-1}g_{2}\in G_{\alpha}$ so $g_{2}=g_{1}h$ for some $h\in G_{\alpha}$
but that contradicts their residence in different cosets, therefore
$g_{1}\left(\alpha\right)\neq g_{2}\left(\alpha\right)$. Applying
these rules and denoting with $G\left(\alpha\right)$ the orbit of
$\alpha$, we get 
\begin{align*}
P\hat{\alpha} & =\left|G\right|^{-1}\sum_{C\in G/G_{a}}\,\,\,\sum_{g\in C}D\left(g\right)\hat{\alpha}\\
 & =\left|G_{a}\right|/\left|G\right|\sum_{C\in G/G_{a}}D\left(g\in C\right)\hat{\alpha}\\
 & =\left|G_{a}\right|/\left|G\right|\sum_{\alpha'\in G\left(\alpha\right)}\hat{\alpha}'.
\end{align*}
It is a well known consequence of the orbit-stabilizer theorem \cite{Armstrong13}
that $\left|G\right|/\left|G_{a}\right|=\left|G\left(\alpha\right)\right|$,
so in fact 
\[
P\hat{\alpha}=\left|G\left(\alpha\right)\right|^{-1}\sum_{\alpha'\in G\left(\alpha\right)}\hat{\alpha}'.
\]
Recalling the general form of a CG projection (\ref{eq:def of P}),
we see that $P$ acts on $\hat{\alpha}$ as the CG projection constructed
from partition of $\mathcal{\mathsf{A}}$ into orbits of $G$. 
\end{proof}
Thus, any group $G$ acting on $\mathcal{\mathsf{A}}$ specifies a
CG associated with the orbits of $G$. Then, if we treat symmetrizations
as a special case of CG, we can specialize the general compatibility
condition of Theorem \ref{thm:PQ=00003DPQP} to this case and express
it in terms of group representations. 
\begin{thm}
\label{thm:SUM(=00005BD(g),Q=00005D)=00003D0 <=00003D> PQ=00003DPQP}
Let 
\begin{equation}
P=\left|G\right|^{-1}\sum_{g\in G_{\mathcal{}}}D\left(g\right)\label{eq:P=00003Dsum_G(D(g))}
\end{equation}
be a symmetrizer with respect to a group $G$, and let $Q$ be a transition
rate matrix. Then, the compatibility condition $PQ=PQP$ is equivalent
to 
\begin{equation}
P\sum_{g\in G}\left[D\left(g\right),Q\right]=0.\label{eq:P*SUM_g(=00005BD,Q=00005D)}
\end{equation}
If in addition $Q=Q^{T}$, then it simplifies to 
\begin{equation}
\sum_{g\in G}\left[D\left(g\right),Q\right]=0.\label{eq:SUM_g(=00005BD,Q=00005D)}
\end{equation}
\end{thm}
\begin{proof}
By definition (\ref{eq:P=00003Dsum_G(D(g))}) of $P$ and the fact
that $P=P^{2}$ we get
\begin{eqnarray*}
PQP & = & \left|G\right|^{-2}\sum_{g,g'\in G}D\left(g\right)QD\left(g'\right)\\
 & = & \left|G\right|^{-2}\sum_{g,g'\in G}D\left(g\right)\left(D\left(g'\right)Q-\left[D\left(g'\right),Q\right]\right)\\
 & = & \left(\left|G\right|^{-1}\sum_{g\in G}D\left(g\right)\right)\left(\left|G\right|^{-1}\sum_{g'\in G}D\left(g'\right)\right)Q\\
 &  & -\left|G\right|^{-1}\left(\left|G\right|^{-1}\sum_{g\in G}D\left(g\right)\right)\left(\sum_{g'\in G}\left[D\left(g'\right),Q\right]\right)\\
 & = & PQ-\left|G\right|^{-1}P\left(\sum_{g'\in G}\left[D\left(g'\right),Q\right]\right),
\end{eqnarray*}
hence the equivalence to (\ref{eq:P*SUM_g(=00005BD,Q=00005D)}). If
in addition $Q=Q^{T}$, then $PQ=PQP$ implies
\[
QP=\left(PQ\right)^{T}=\left(PQP\right)^{T}=PQP=PQ,
\]
that is $\left[P,Q\right]=0$. And also $\left[P,Q\right]=0$ implies
$PQ=PQP$ hence the equivalence to (\ref{eq:SUM_g(=00005BD,Q=00005D)}).
\end{proof}
As was pointed out after corollary \ref{cor:total rates must be uniform},
CGs that are compatible with dynamics form an equitable partition
of the graph specified by the weighted adjacency matrix $Q$. In \cite{Barrett17},
graph automorphism symmetries (permutations of vertices that commute
with the weighted adjacency matrix $Q$) were used to single out equitable
partitions with their orbits. Theorem \ref{thm:SUM(=00005BD(g),Q=00005D)=00003D0 <=00003D> PQ=00003DPQP}
confirms this, as Eq. (\ref{eq:P*SUM_g(=00005BD,Q=00005D)}) trivially
holds for all groups that satisfy $\left[D\left(g\right),Q\right]=0$
for all $g$. However, Theorem \ref{thm:SUM(=00005BD(g),Q=00005D)=00003D0 <=00003D> PQ=00003DPQP}
(together with Corollary \ref{cor:total rates must be uniform}) implies
that there is a broader set of symmetries, beyond automorphisms, that
specify equitable partitions with their orbits. These are the groups
that comply with Eq. (\ref{eq:P*SUM_g(=00005BD,Q=00005D)}) or (\ref{eq:SUM_g(=00005BD,Q=00005D)}).

\subsection{Example: Glauber-Ising Model}

The system that we study here is a 1D classical spin lattice with
periodic boundary conditions, i.e., an Ising spin chain. We will see
that compatible CGs of this system are not so obvious (and the obvious
ones are not compatible). We will overcome this difficulty by putting
to use the considerations of symmetry developed in the previous section. 

Each of the sites in the Ising chain can be in one of two states $\left\{ \pm1\right\} $.
A micro state of the lattice of length $N$ is an $N$-component binary
vector $\sigma\in\left\{ \pm1\right\} ^{N}$, and the state-space
consists of $2^{N}$ micro states $\left\{ \sigma_{i}\right\} _{i=1}^{2^{N}}$.
The internal energy of a micro state $\sigma$ is 
\[
E\left(\sigma\right)=-J\underset{x=1}{\overset{N}{\sum}}\sigma^{\left(x\right)}\sigma^{\left(x+1\right)},
\]
where $J>0$ is the local interaction energy, and $\sigma^{\left(x\right)}$
is the sign of site $x$.

The Glauber-Ising model, proposed by Glauber in \cite{Glauber63},
is a model of dynamics for an Ising spin chain that interacts thermally
with its environment. According to the model, the micro state of the
system evolves by transitions caused by single spin flips. The transition
rate depends on whether the energy increases, decreases or stays the
same
\[
r\left(\sigma\longmapsto\sigma'\right):=\begin{cases}
1-\gamma\left(T\right) & :\,E\left(\sigma'\right)>E\left(\sigma\right)\\
1 & :\,E\left(\sigma'\right)=E\left(\sigma\right)\\
1+\gamma\left(T\right) & :\,E\left(\sigma'\right)<E\left(\sigma\right),
\end{cases}
\]
where $\gamma\left(T\right)$ is a positive, temperature-dependent,
parameter. This model simply states that transitions that increase
$E$ happen at a slower rate than the ones that decrease $E$, and
this rate difference is additively modified by the temperature through
$\gamma\left(T\right)$. All three rates are in units that normalize
the middle rate to $1$. The parameter $\gamma\left(T\right)$ can
then be derived by demanding\textit{ }detailed balance condition in
equilibrium, which leads to $\gamma\left(T\right)=\tanh\left[\frac{2J}{T}\right]$
(see \cite{Glauber63} for details). With the rate function $r\left(\sigma\longmapsto\sigma'\right)$
we can in principle construct the $2^{N}\times2^{N}$ transition rate
matrix $Q$. 

In order to understand how this dynamical system can be coarse-grained,
we look at the case of $N=3$ first. Instead of writing down the matrix
$Q$ explicitly, we describe the dynamics as a random walk on the
graph depicted in fig \ref{fig:spinchain random walk}(a). 
\begin{figure}[t]
\begin{raggedright}
(a) %
\noindent\begin{minipage}[t]{1\columnwidth}%
\includegraphics[width=1\columnwidth]{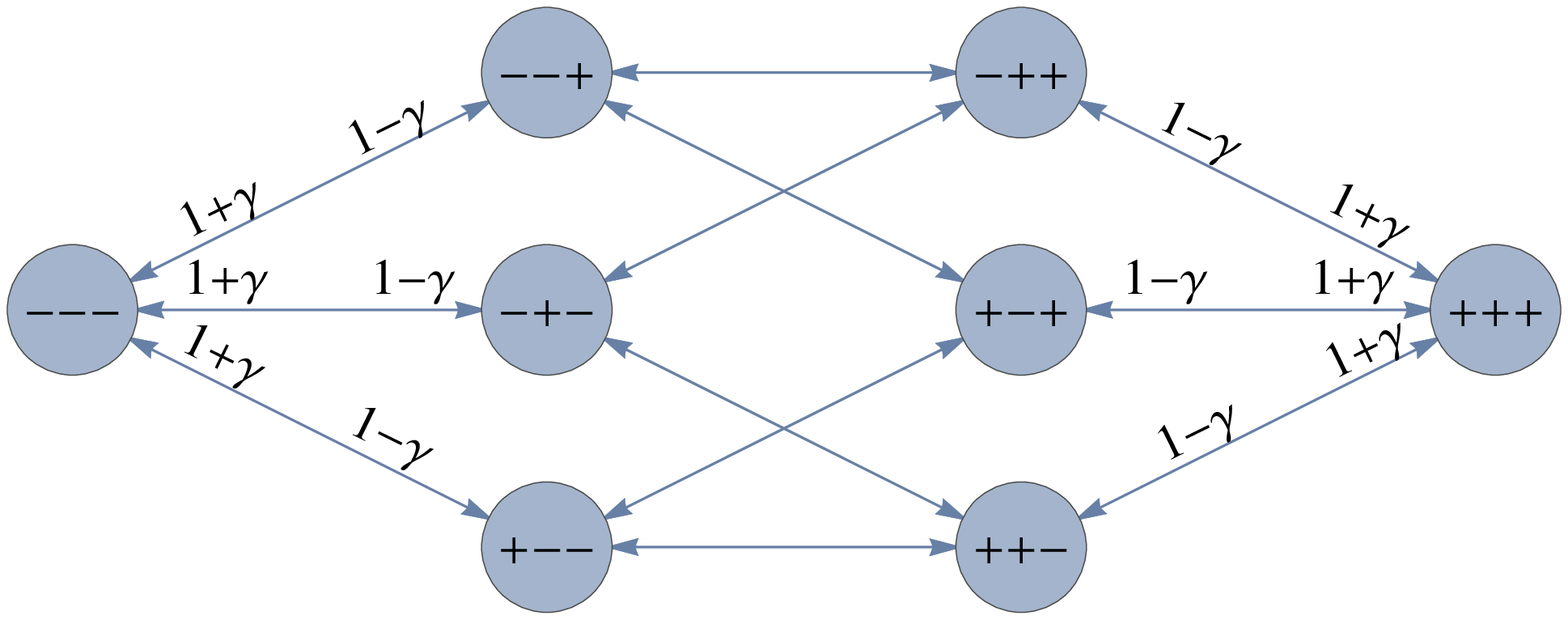}%
\end{minipage}
\par\end{raggedright}
\raggedright{}(b) %
\noindent\begin{minipage}[t]{1\columnwidth}%
\includegraphics[width=1\columnwidth]{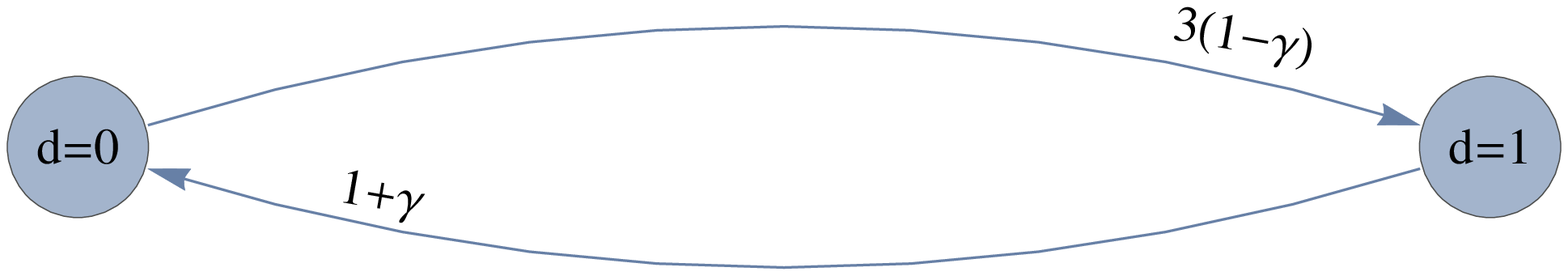}%
\end{minipage}\caption{\label{fig:spinchain random walk}(a) Random walk graph for a 3-spin
Ising chain. Glauber's transition rates are explicitly specified on
the edges only where it differs from 1. Transition rate 1 is implied
for unlabeled edges. (b) Coarse-graining of the state-space of 3-spin
Ising chain with respect to the total number of domains $d$. Effective
transition rates between states with $d=0$ and $d=1$ specified on
the edges. }
\end{figure}
 The signs $\pm$ stand for $\pm1$ and transition rates are explicitly
specified only where they differ from $1$. The total transition rate
from each of the ground states $\left(+++\right)$, $\left(---\right)$,
to the bulk of excited states (the ones in the middle) sum to $3\left(1-\gamma\right)$.
In the opposite direction, from excited to the ground, there is only
one transition for each of the excited states, and it is at the same
rate $1+\gamma$. Corollary \ref{cor:total rates must be uniform}
then implies that coarse-graining this system with respect to the
energy levels is compatible with dynamics. Instead of energy we can
count the number of domains $d$ (defined as half the number of intervals
in the chain that differ in sign from their surroundings) which is
a proxy variable for energy as seen from the relation $E=-J\left(N-4d\right)$.
In fig. \ref{fig:spinchain random walk}(b) we see the reduced state-space,
coarse-grained by blocking together micro states that have the same
energy / number of domains.

The reduced transition rate matrix 
\[
\tilde{Q}=\begin{pmatrix}-3\left(1-\gamma\right) & 1+\gamma\\
3\left(1-\gamma\right) & -1-\gamma
\end{pmatrix}
\]
generates a random walk in the state-space of the number-of-domains
variable $d\in\left\{ 0,1\right\} $. The eigenvalues of $\tilde{Q}$
are $\lambda_{0}=0$ and $\lambda_{1}=2\left(\gamma-2\right)$ correspond
to the eigenvectors 
\[
v_{0}=\frac{1}{2\left(2-\gamma\right)}\begin{pmatrix}1+\gamma\\
3\left(1-\gamma\right)
\end{pmatrix}\,\,\,\,\,\,\,\,\,\,\,\,\,\,v_{1}=\begin{pmatrix}-1/2\\
1/2
\end{pmatrix}.
\]
Since $e^{t\tilde{Q}}v_{0}=v_{0}$, the vector $v_{0}$ is the steady/equilibrium
state of the system, and its components are the probabilities of finding
the system in one of the energy levels when the system is in equilibrium.
So, for example, the probability of finding this system in the excited
state in equilibrium is 
\[
\Pr\left(E=J\right)=\frac{3\left(1-\tanh\left[\frac{2J}{T}\right]\right)}{2\left(2-\tanh\left[\frac{2J}{T}\right]\right)}=\frac{3}{3+e^{4J/T}}.
\]
This expression agrees with the standard calculation of Boltzmann's
factor and partition function. 

In addition to recovering equilibrium properties from $\tilde{Q}$,
we can also learn something about out-of-equilibrium behavior. Since
$v_{0}$ is a normalized probability vector, we can always add to
it a fraction of $v_{1}$ to get any other normalized probability
vector. So, any initial state $p_{in}$ can be written as $v_{0}+rv_{1}$,
where $r$ is a real parameter. Its time evolution is then given by
\[
p\left(t\right)=e^{t\tilde{Q}}p_{in}=v_{0}+re^{t\lambda_{1}}v_{1}.
\]
Note that $\lambda_{1}<-2$ because $\gamma=\tanh\left[\frac{2J}{T}\right]<1$,
therefore every initial state relaxes to equilibrium $v_{0}$ by exponentially
suppressing $v_{1}$. This means that the characteristic relaxation
time for this system is
\[
-\frac{1}{\lambda_{1}}=-\frac{1}{2\left(\gamma-2\right)}=\frac{1}{2}\left(\frac{1+e^{4J/T}}{3+e^{4J/T}}\right).
\]

Even though the intuitive CG with respect to the energy levels is
compatible with dynamics for $N=3$, it is not true in general (see
the case of $N=4$ bellow). In the general case we will look for a
group that complies with the compatibility condition of Theorem \ref{thm:SUM(=00005BD(g),Q=00005D)=00003D0 <=00003D> PQ=00003DPQP},
and take its orbits to be the compatible CG blocks. 

From the $N=3$ case we see that the group $\mathbb{Z}_{3}$ of lattice
translations generates orbits that coincide with columns in fig \ref{fig:spinchain random walk}(a).
If we complement this group with $\mathbb{Z}_{2}$ of global spin
flips then $\mathbb{Z}_{3}\times\mathbb{Z}_{2}$ generates 2 orbits
that coincide with the blocks of $d=0$ and $d=1$. Since these blocks
are compatible with dynamics, we conjecture that for general $N$
the orbits of $G=\mathbb{Z}_{N}\times\mathbb{Z}_{2}$ (translations
and global flips) are compatible with dynamics.

To prove that, we note that the transition rate matrix can be decomposed
as a sum of $N$ matrices $Q=\sum_{x=1}^{N}Q^{\left(x\right)}$, where
each $Q^{\left(x\right)}$ generates transitions restricted to flips
of site $x$. If $D\left(x\right)$ represents the action of lattice
translations by $x$ sites, and $D\left(x\right)D\left(y\right)=D\left(x+y\right)$,
then we get
\[
D\left(y\right)Q^{\left(x\right)}=D\left(y\right)Q^{\left(x\right)}D\left(-y\right)D\left(y\right)=Q^{\left(y+x\right)}D\left(y\right).
\]
Therefore,
\[
D\left(y\right)Q=\sum_{x\in\mathbb{Z}_{N}}Q^{\left(y+x\right)}D\left(y\right)=QD\left(y\right),
\]
that is $\left[D\left(y\right),Q\right]=0$. Each local spin flip
generator $Q^{\left(x\right)}$ also commutes with the global spin
flip action $D\left(f\right)$, therefore $\left[D\left(g\right),Q\right]=0$
for all $g\in G$. Since $G$ is a symmetry group of $Q$, it satisfies
the compatibility condition of Theorem \ref{thm:SUM(=00005BD(g),Q=00005D)=00003D0 <=00003D> PQ=00003DPQP},
and we can coarse-grain this dynamical system by blocking together
the states that belong to the same orbit of $G$. 

For $N=4$ the orbits of $G$ are
\[
\begin{array}{lc}
d=0 & \left\{ \begin{array}{lc}
\textrm{orbit}\,1 & \left\{ \begin{array}{cc}
\left(++++\right) & \left(----\right)\end{array}\right.\end{array}\right.\\
\\
d=1 & \left\{ \begin{array}{ll}
\textrm{orbit}\,2 & \left\{ \begin{array}{cl}
\left(+---\right) & \left(-+++\right)\\
\left(-+--\right) & \left(+-++\right)\\
\left(--+-\right) & \left(++-+\right)\\
\left(---+\right) & \left(+++-\right)
\end{array}\right.\\
\\
\textrm{orbit}\,3 & \left\{ \begin{array}{cc}
\left(++--\right) & \left(--++\right)\\
\left(-++-\right) & \left(+--+\right)
\end{array}\right.
\end{array}\right.\\
\\
d=2 & \left\{ \begin{array}{lc}
\textrm{orbit}\,4 & \left\{ \begin{array}{cc}
\left(+-+-\right) & \left(-+-+\right),\end{array}\right.\end{array}\right.
\end{array}
\]
so orbit 1 coincides with the lowest energy level, orbits 2 and 3
together form the first excited level, and orbit 4 coincides with
the second excited level. Transition rates between the orbits are
shown in fig. \ref{fig:spinchain random walk-2}.
\begin{figure}[t]
\raggedright{}\includegraphics[width=1\columnwidth]{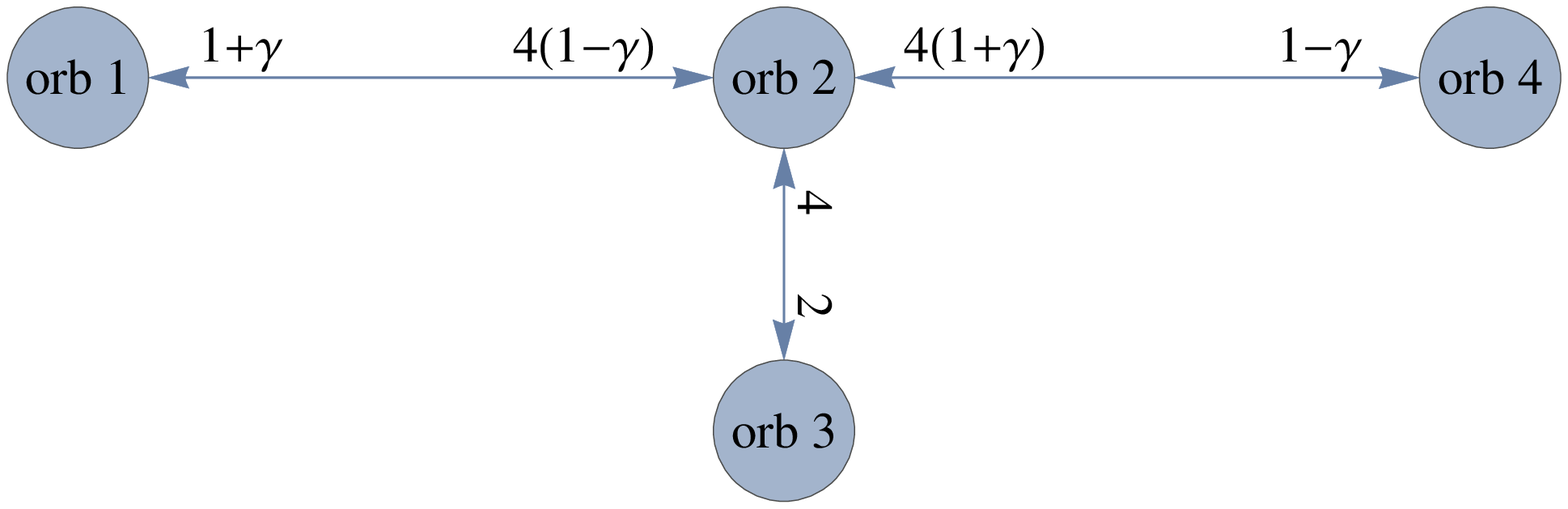}\caption{\label{fig:spinchain random walk-2}Random walk graph for a 4-spin
state-space coarse-grained with respect to orbits of translations
and global spin flips.}
\end{figure}
 Note that the rate of transitions from orbit 2 to the neighboring
energy levels are $1\pm\gamma$ but from orbit 3 it is $0$, because
no single spin flip can change the energy. In general, the total transition
rates to the neighboring energy levels are not constant over the states
in each energy level. That is why CG by energy levels is not compatible
with dynamics. For $N>3$ energy levels happen to be too coarse to
be compatible, but the orbits of $G$ are fine enough. 

It is curious to note that in the thermodynamic limit $N\longrightarrow\infty$,
each energy level consists almost entirely of states that have the
same total transition rate to the neighboring levels. The non-uniformity
of rates over the energy levels is then suppressed, and transition
rates between energy levels can be approximately defined, but this
analysis is beyond our scope here.

\subsection{\label{subsec:Partial-Subsystems-and}Partial Subsystems and Bipartitions}

The concept of a partial subsystems that we will define here is a
natural byproduct of the CG discussion. It follows from the observation
that for a bipartite system, marginalization of one of its subsystems
is a special case of CG. If so, it is natural to ask whether any CG
can be viewed as marginalization of some kind of subsystem. The answer
is yes if one is willing to stretch the meaning of subsystem. This
leads us to the definition of a \textit{partial }subsystem. In the
context of classical CG it is hardly worth the effort but the goal
here is to prepare the ground for quantum CG. The raw notion of CG
does not land naturally in quantum theory, but it easily sneaks in
as marginalization of a partial subsystem. 

Consider the state-space of a composite system $AB$ that is the cartesian
product $\mathsf{C}:=\mathsf{A}\times\mathsf{B}$, where $\mathsf{A}:=\left\{ \alpha_{i}\right\} _{i=1}^{\mathsf{\left|A\right|}}$
and $\mathsf{B}:=\left\{ \beta_{k}\right\} _{k=1}^{\mathsf{\left|B\right|}}$
are the state-spaces of individual subsystems (now both $\alpha_{i}$
and $\beta_{k}$ refer to micro states). Probabilistic states $p_{\mathsf{C}}$
live in the vector space $\mathbb{R}^{\mathsf{C}}$, and marginalization
of subsystem $A$ is given by the map $\mathcal{M}:\mathbb{R}^{\mathsf{A}\times\mathsf{B}}\longrightarrow\mathbb{R}^{\mathsf{B}}$
which accounts for our ignorance of system $A$. The map $\mathcal{M}$
is defined by the relation
\[
p_{\mathcal{\mathsf{B}}}\left(\beta\right)=\sum_{\alpha\in\mathsf{A}}p_{\mathcal{\mathsf{C}}}\left(\alpha\times\beta\right).
\]
The resemblance between this equation and Eq. (\ref{eq: p_B(b)=00003DSum_a_(p_S(a))})
is obvious. If we partition $\mathsf{C}$ into blocks $\left\{ \alpha\times\beta\right\} _{\alpha\in A}^{\mathsf{}}$
for each $\beta\in\mathsf{B}$, and slightly abuse the notation by
also referring to each block with $\beta$, then CG as defined by
Eq. (\ref{eq: p_B(b)=00003DSum_a_(p_S(a))}) marginalizes $A$, and
the action of CG matrix $M$ is identical to that of map $\mathcal{M}$. 

Marginalization is a special case of CG where all blocks are of the
same size. In general this is not the case, but if we happen to partition
a system (not necessarily composite) into blocks of equal size, we
can think about it as a composite system of two \textit{virtual} subsystems.
Consider the state-space $\mathsf{C}:=\left\{ \gamma_{ik}\right\} $
where the indices $i=1...n$ and $k=1...m$ define a partition of
$\mathsf{C}$ into $m$ blocks of $n$ elements each. We then can
imagine systems $\mathsf{A}:=\left\{ \alpha_{i}\right\} _{i=1}^{n}$
and $\mathsf{B}:=\left\{ \beta_{k}\right\} _{k=1}^{m}$ and identify
the states $\gamma_{ik}\equiv\alpha_{i}\times\beta_{k}$ so $\mathsf{C}\cong\mathsf{A}\times\mathsf{B}$.
Such subsystems are commonly known as \textit{virtual subsystems}.
Coarse-graining with respect to this partition is effectively a marginalization
of the virtual subsystem $A$.

\noindent In general we can specify any partition of $\mathsf{C}:=\left\{ \gamma_{ik}\right\} $
by assignment of indices $i,k$, where $k$ refers to the block and
$i$ to the relative position of elements inside the block. It is
convenient to order the blocks by descending block size and arrange
the elements $\left\{ \gamma_{ik}\right\} $ into what we call a \textit{bipartition
table.} 
\begin{table}[H]
\centering{}%
\begin{tabular}{|c|c|c|c|c|ccc|c|}
\cline{1-7} \cline{9-9} 
$\gamma_{1,1}$ & ... & $\gamma_{1,k}$ & ... & $\gamma_{1,w_{i}}$ & \multicolumn{1}{c|}{...} & \multicolumn{1}{c|}{$\gamma_{1,m}$} & $\shortrightarrow$ & $\alpha_{1}$\tabularnewline
\cline{1-7} \cline{9-9} 
$\vdots$ & $\vdots$ & $\vdots$ & $\vdots$ & $\vdots$ & \multicolumn{1}{c|}{$\vdots$} & \multicolumn{1}{c|}{$\vdots$} &  & $\vdots$\tabularnewline
\cline{1-7} \cline{9-9} 
$\gamma_{i,1}$ & ... & $\gamma_{i,k}$ & ... & $\gamma_{i,w_{i}}$ &  &  & $\shortrightarrow$ & $\alpha_{i}$\tabularnewline
\cline{1-5} \cline{9-9} 
$\vdots$ & $\vdots$ & $\vdots$ & $\vdots$ & $\vdots$ &  &  &  & $\vdots$\tabularnewline
\cline{1-5} \cline{9-9} 
$\gamma_{h_{k},1}$ & ... & $\gamma_{h_{k},k}$ & \multicolumn{1}{c}{} & \multicolumn{1}{c}{} &  &  & $\shortrightarrow$ & $\alpha_{h_{k}}$\tabularnewline
\cline{1-3} \cline{9-9} 
$\vdots$ & $\vdots$ & \multicolumn{1}{c}{} & \multicolumn{1}{c}{} & \multicolumn{1}{c}{} &  &  &  & $\vdots$\tabularnewline
\cline{1-2} \cline{9-9} 
$\gamma_{n,1}$ & ... & \multicolumn{1}{c}{} & \multicolumn{1}{c}{} & \multicolumn{1}{c}{} &  &  & $\shortrightarrow$ & $\alpha_{n}$\tabularnewline
\cline{1-2} \cline{9-9} 
\multicolumn{1}{c}{$\shortdownarrow$} & \multicolumn{1}{c}{} & \multicolumn{1}{c}{$\shortdownarrow$} & \multicolumn{1}{c}{} & \multicolumn{1}{c}{$\shortdownarrow$} &  & $\shortdownarrow$ & \multicolumn{1}{c}{} & \multicolumn{1}{c}{}\tabularnewline
\cline{1-7} 
$\beta_{1}$ & ... & $\beta_{k}$ & ... & $\beta_{w_{i}}$ & \multicolumn{1}{c|}{...} & \multicolumn{1}{c|}{$\beta_{m}$} & \multicolumn{1}{c}{} & \multicolumn{1}{c}{}\tabularnewline
\cline{1-7} 
\end{tabular}\caption{\textit{\label{tab:Bipartition-Table}Bipartition Table. Arrows point
toward the associated states of partial subsystems. }}
\end{table}
The columns of this table correspond to CG blocks. If all blocks are
of the same size, then the table is rectangular and the set of columns
(rows) is associated with states of the virtual subsystems $B$ ($A$),
as indicated by the arrows. When the blocks are not all of the same
size, the range of indices $i,k$ is not independent from each other.
If $k=1...m$ for a fixed $m$, then $i=1...h_{k}$ where $h_{k}$
is the size of block $k$ (height of column $k$). We can always invert
the dependence so if $i=1...n$, then $k=1...w_{i}$ (width of row
$i$). Even when the table is not rectangular, we can still associate
the columns (rows) with states of fictitious system $B$ ($A$), and
identify $\gamma_{ik}\equiv\alpha_{i}\times\beta_{k}$ as composite
states. We call such fictitious subsystems \textit{partial subsystems.
}What sets them apart from virtual subsystems\textit{ }is the fact
that certain combinations of states are not allowed. The injective
map $V:\mathsf{C}\longrightarrow\mathsf{A}\times\mathsf{B}$, which
assigns elements of $\mathsf{C}$ into the bipartition table, will
be called a \textit{partial bipartition map} (the bipartition is not
partial if $V$ is bijective).

Marginalization of a partial subsystem is essentially the same procedure
as marginalization of the non-partial subsystem. We sum the probabilities
over the rows or the columns of the bipartition table and assign them
to reduced states. The fact that some combinations of composite states
are not allowed simply means that they contribute nothing to the sums.
To make this more rigorous, consider the partial bipartition map $V:\mathsf{C}\longrightarrow\mathsf{A}\times\mathsf{B}$.
By applying $V$ on the corresponding basis of $\mathbb{R}^{\mathsf{C}}$
we get the isomorphic embedding $\mathcal{V}:\mathbb{R}^{\mathsf{C}}\longrightarrow\mathbb{R}^{\mathsf{A}\times\mathsf{B}}$.
The map $\mathcal{V}$ embeds probabilistic states of $C$ into a
subspace of probabilistic states of $AB$, spanned by the allowed
combinations of states. Marginalizing with $\mathcal{M}:\mathbb{R}^{\mathsf{A}\times\mathsf{B}}\longrightarrow\mathbb{R}^{\mathsf{B}}$
after embedding with $\mathcal{V}$ defines the marginalization of
the partial subsystem $A$ 
\[
\mathcal{M}\mathcal{V}:\mathbb{R}^{\mathsf{C}}\longrightarrow\mathbb{R}^{\mathsf{B}}.
\]
Intuitively, the map $\mathcal{V}$ completes the missing blocks of
the bipartition table to make it rectangular and assigns zero probability
to the missing states. Then $\mathcal{M}$ sums the probabilities
over the columns and assigns them to the associated states of $B$.
Thus, $\mathcal{M}\mathcal{V}$ sums the probabilities over the columns
of the bipartition table which means that $\mathcal{M}\mathcal{V}$
implements a CG of $\mathsf{C}$ according to the blocks defined by
the columns of the table.

To recap, every system admits a partial bipartition into partial subsystems.
Partial bipartition is defined by the shape of the bipartition table
and the assignment of elements into it. Columns (rows) of the bipartition
table are associated with states of partial subsystem $B$ ($A$).
We saw that marginalization of a partial subsystem is equivalent to
CG over the columns. The fact that every CG is a marginalization of
a partial subsystem is easy to see: just arrange the CG blocks into
columns of a bipartition table with arbitrary ordering inside the
blocks. Thus, CG can be equivalently defined as marginalization of
a partial subsystem. This definition has a bit of extra structure
that is not strictly necessary for classical CG. The extra structure
is in the bipartition table which assigns order to elements inside
the blocks (columns) and it is irrelevant if we simply sum over them. 

It turns out that in the quantum version of marginalization \textendash{}
the partial trace \textendash{} this ordering makes a difference.
This also explains why we could not directly export CG into quantum
theory from the basic definitions of section \ref{subsec:Basic-Definitions}.
While partition into blocks provides enough structure to specify a
CG for the classical state-space, we need the extra structure of the
bipartition table to specify a CG for the quantum state-space.

\section{Quantum Coarse-Graining}

The fundamental feature of quantum systems that sets them apart from
their classical analogues is the superposition principle \cite{Dirac81}.
Therefore, for a notion of coarse-graining to be truly ``quantum'',
we must embrace the superposition principle and allow the possibility
of reducing superpositions of micro states into superpositions of
macro states. Section \ref{subsec: Formal Definitions} is dedicated
to the formal definition of such notion. 

Although the definition of quantum CG (QCG) is quite simple, the interpretation
requires some effort. We will show that just like in the classical
case, QCG is a manifestation of restricted access to observables.
The main technical result behind it is the definition of \textit{bipartition
operators}. In section \ref{subsec:Examples:-Operationally-Motivate}
we will demonstrate this formalism in some special cases.

Section \ref{subsec:quantum Compatibility-with-Dynamics} addresses
the question of reducibility of dynamics. The problem is formulated
in terms of a compatibility condition between a QCG scheme and a generator
of dynamics. We will see that the general condition for compatibility
can be derived and presented in exactly the same form as in the classical
case. This result is then specialized to unitary quantum dynamics
by Theorem \ref{thm:=00005BH,S=00005D in span=00007BS_kl=00007D}.

In Section \ref{subsec:Coarse Graining and Symmetries} we focus on
symmetries and the associated QCGs. Symmetry considerations have been
fundamental in the development of many important ideas in physics:
from Emmy Noether's seminal work \cite{Noether18} relating conserved
quantities to the symmetries of dynamics; to the modern applications
in subjects ranging from crystallography \cite{Cornwell} to error
avoidance strategies in quantum computers \cite{Lidar13}. Many of
the applications of symmetries, including Noether's Theorem, are concerned
with their implications for dynamical processes. Therefore, the analysis
of reducibility of dynamics would not be complete without introducing
symmetry considerations. The main result that we derive in that regard
is Theorem \ref{thm:comp. of symmetrization with dynamocs}, where
we specialize the compatibility condition to QCGs by symmetrization.
This compatibility condition turns out to be a relaxed version of
symmetry of dynamics, where the commutator $\left[U\left(g\right),H\right]$
may not vanish, but it has to belong to the operator algebra spanned
by the group.

We end this section with an example of continuous-time quantum walk
on a tree. It shows that symmetries can be used to reduce the dynamics
even when they do not commute with the Hamiltonian.

\subsection{Formal Definitions\label{subsec: Formal Definitions}}

The difficulty with direct adoption of the classical notion of CG,
in the quantum setting, arises because the classical notion is ignorant
of the possibility of superpositions between the states. Consider
a finite dimensional Hilbert space as a quantized version of the classical
state-space, where micro states $\alpha_{i}$ were promoted to orthonormal
basis $\ket{\alpha_{i}}$. If we partition the micro states into blocks
$\left\{ \ket{\alpha_{i}}\right\} _{i\in b}$, it may still make sense
to say that all states $\ket{\psi_{b}}$ that belong to the span of
block $b$ are indistinguishable and should be reduced as $\ket{\psi_{b}}\longmapsto\ket b$.
However, if we look at superpositions such as $\ket{\psi}=\ket{\psi_{b}}+\ket{\psi_{b'}}$,
this CG reduction is not consistently defined. If we naively suggest
that $\ket{\psi}\longmapsto\ket b+\ket{b'}$, then we can always write
the same state differently $\ket{\psi}=e^{i\varphi}\ket{\tilde{\psi}_{b}}+\ket{\psi_{b'}}$
where $\ket{\tilde{\psi}_{b}}:=e^{-i\varphi}\ket{\psi_{b}}$, and
get a different reduced state $\ket{\psi}\longmapsto e^{i\varphi}\ket b+\ket{b'}$. 

Reduction of coherence terms between the blocks is simply undefined
by the classical CG procedure. If we insist on using the classical
notion as it is, the only sensible approach is to discard the coherence
terms altogether. That is, the reduction of $\ket{\psi}=\ket{\psi_{b}}+\ket{\psi_{b'}}$
should be $\ket{\psi}\longmapsto\ket b\bra b+\ket{b'}\bra{b'}$. Such
CG of quantum states is consistently defined, but it is not truly
quantum. 

For the truly quantum notion of CG we have to consistently account
for coherence terms between the CG blocks. In order to do that, we
will adopt a different perspective. Recall that classical CG was eventually
understood as marginalization of a (partial) subsystem. This observation
is key, because it shifts the focus from blocks and states to subsystems.
Thus, quantum coarse-graining can be introduced as quantum marginalization
of a partial subsystem. Since the notion of quantum marginalization
\textendash{} the partial trace map \textendash{} already exists,
all we have left is to identify partial subsystems in the quantum
setting.

It should be noted that mathematically equivalent definitions of the
QCG map have been presented in \cite{Faist16} and \cite{Duarte17}.
Our definition is different in that it is derived by a direct analogy
with the classical case. Furthermore, we will expand on the formal
structure of QCG by identifying \textit{bipartition operators} as
key mathematical objects and \textit{bipartition tables} as their
graphical representation. We will also provide QCG with operational
meaning. 

\subsubsection{Partial Subsystems and Bipartitions}

Consider a finite dimensional physical Hilbert space $\mathcal{H}$.
The choice of orthonormal basis $\left\{ \ket{\gamma_{ik}}\right\} $
and their arrangement into a bipartition table \ref{tab:Bipartition-Table}
constitutes a partial bipartition of $\mathcal{H}$. The auxiliary
Hilbert space $\mathcal{H}^{A}$ ($\mathcal{H}^{B}$) of the partial
subsystem $A$ ($B$) is formally defined as the span of row kets
$\left\{ \left|\alpha_{i}\right\rangle \right\} _{i=1}^{n}$ (column
kets $\left\{ \left|\beta_{k}\right\rangle \right\} _{k=1}^{m}$)
as illustrated in table \ref{tab:Bipartition-Table}. The physical
Hilbert space $\mathcal{H}$ can now be isometrically embedded into
the subspace of $\mathcal{H}^{AB}:=\mathcal{H}^{A}\otimes\mathcal{H}^{B}$
with the map $V:\left|\gamma_{i,k}\right\rangle \longmapsto\left|\alpha_{i}\right\rangle \left|\beta_{k}\right\rangle $.
For every $\left|\gamma_{i,k}\right\rangle \in\mathcal{H}$ there
is a matching pair $\left|\alpha_{i}\right\rangle \left|\beta_{k}\right\rangle \in\mathcal{H}^{AB}$,
but not vice versa. The extra pairs in $\mathcal{H}^{AB}$ that don't
have a match in $\mathcal{H}$ correspond to the missing elements
of the bipartition table that would complete it to a rectangular form.
The case where the chosen bipartition table of $\mathcal{H}$ is rectangular,
so $\mathcal{H}\cong\mathcal{H}^{AB}$, is the case where $\mathcal{H}^{A}$
and $\mathcal{H}^{B}$ were identified by \cite{Zanardi01} as \textit{virtual}
subsystems. The construction here is more general, therefore we refer
to such subsystems as \textit{partial} subsystems. 

In the following it will be useful to express the isometry $V$ in
two complementary forms 
\begin{equation}
V=\sum_{k=1}^{m}V_{k}^{A}\otimes\left|\beta_{k}\right\rangle =\sum_{i=1}^{n}\left|\alpha_{i}\right\rangle \otimes V_{i}^{B}.\label{eq:def of V}
\end{equation}
The partial isometries 
\begin{equation}
V_{k}^{A}:=\sum_{i=1}^{h_{k}}\left|\alpha_{i}\right\rangle \left\langle \gamma_{i,k}\right|\,\,\,\,\,\,\,\,\,\,\,\,\,\,\,\,\,\,V_{i}^{B}:=\sum_{k=1}^{w_{i}}\left|\beta_{k}\right\rangle \left\langle \gamma_{i,k}\right|\label{eq:def of V^A and V_B}
\end{equation}
map the individual columns (rows) of the bipartition table into $\mathcal{H}^{A}$
($\mathcal{H}^{B}$). 

\subsubsection{Quantum Coarse-Graining Map}

Once the partial subsystem $A$ is identified, QCG is defined as the
map that traces out $A$. Since the partial trace map $tr_{A}$ acts
on operators, we have to elevate the action of the isometry $V$ to
operators as well, thus defining $\mathcal{V}\left(\cdot\right):=V\left(\cdot\right)V^{\dagger}$.
Then, the composition 
\[
tr_{\left(A\right)}:=tr_{A}\circ\mathcal{V}
\]
defines the QCG map $tr_{\left(A\right)}$. Since both components
of this composition are completely positive trace preserving (CPTP)
maps, QCG map $tr_{\left(A\right)}$ reduces proper quantum states
to proper quantum states \cite{Nilsen=000026Chuang}. Operator sum
representation of $tr_{\left(A\right)}$ can be obtained by expressing
$V$ in the second form of Eq. (\ref{eq:def of V}) and applying $tr_{A}$
\begin{align*}
tr_{\left(A\right)}\left(\rho\right) & =tr_{A}\left(V\rho V^{\dagger}\right)=\sum_{i=1}^{n}V_{i}^{B}\rho\,V_{i}^{B\dagger}.
\end{align*}
Reduction with $tr_{\left(A\right)}$ maps the density matrices between
the operator spaces as
\[
tr_{\left(A\right)}:\mathcal{B}\left(\mathcal{H}\right)\longrightarrow\mathcal{B}\left(\mathcal{H}^{A}\otimes\mathcal{H}^{B}\right)\longrightarrow\mathcal{B}\left(\mathcal{H}^{B}\right),
\]
so the partial subsystem $B$ embodies the reduced, coarse-grained,
state-space. 

The choice of notation $tr_{\left(A\right)}$ for the QCG map is justified
by its action on the matrix elements in the bipartition basis $\left|\gamma_{i,k}\right\rangle $

\begin{align}
tr_{\left(A\right)}:\left|\gamma_{i,k}\right\rangle \left\langle \gamma_{j,l}\right| & \longmapsto\delta_{ij}\left|\beta_{k}\right\rangle \left\langle \beta_{l}\right|.\label{eq:Tr_(a) acting on matrix element}
\end{align}
So it traces over the indices $i$, $j$ as if they label basis elements
of a proper subsystem (the bracketed subscript $_{\left(A\right)}$,
as opposed to the unbracketed one $_{A}$, refers to the fact that
it traces over a \textit{partial} subsystem). 

As an illustration, consider the 6 dimensional Hilbert space $\mathcal{H}$
spanned by the orthonormal basis $\{\ket s\}$ for $s=1,...,6$. A
partial bipartition of $\mathcal{H}$ is chosen s.t. in the basis
$\{\ket s\}$ it is specified by the bipartition table
\begin{center}
\begin{tabular}{|c|c|c}
\hline 
$1$ & $2$ & \multicolumn{1}{c|}{$3$}\tabularnewline
\hline 
$4$ & $5$ & \tabularnewline
\cline{1-2} 
$6$ & \multicolumn{1}{c}{} & \tabularnewline
\cline{1-1} 
\end{tabular}
\par\end{center}

\noindent We will now use the notation $\ket{\gamma_{i,k}}$ to refer
to the same elements $\ket s$ by their row / column indices; for
example $\ket 4\equiv\ket{\gamma_{2,1}}$. 

An arbitrary pure state can then be written as $\left|\psi\right\rangle =\left|\psi_{1}\right\rangle +\left|\psi_{2}\right\rangle +\left|\psi_{3}\right\rangle $,
where each unnormalized state $\left|\psi_{i}\right\rangle $ is the
support of $\left|\psi\right\rangle $ on the row $i$ 
\begin{align*}
\left|\psi_{1}\right\rangle  & :=c_{11}\ket{\gamma_{1,1}}+c_{12}\ket{\gamma_{1,2}}+c_{13}\ket{\gamma_{1,3}}\\
\left|\psi_{2}\right\rangle  & :=c_{21}\ket{\gamma_{2,1}}+c_{22}\ket{\gamma_{2,2}}\\
\left|\psi_{3}\right\rangle  & :=c_{33}\ket{\gamma_{3,1}}.
\end{align*}
Applying Eq. (\ref{eq:Tr_(a) acting on matrix element}) on the element
$\left|\psi_{i}\right\rangle \left\langle \psi_{j}\right|$ we get
\[
tr_{\left(A\right)}\left(\left|\psi_{i}\right\rangle \left\langle \psi_{j}\right|\right)=\delta_{ij}\sum_{k,l}c_{ik}\overline{c_{il}}\left|\beta_{k}\right\rangle \left\langle \beta_{l}\right|.
\]
Then if we present the density matrix $\rho:=\left|\psi\right\rangle \left\langle \psi\right|$
in the bipartition basis ordered by their appearance in the bipartition
table (read from left to right and top to bottom), the action of $tr_{\left(A\right)}$
is 
\[
\begin{array}{c}
\begin{pmatrix}{\color{red}\rho_{11}} & {\color{red}\rho_{12}} & {\color{red}\rho_{13}} & \rho_{14} & \rho_{15} & \rho_{16}\\
{\color{red}\rho_{21}} & {\color{red}\rho_{22}} & {\color{red}\rho_{23}} & \rho_{24} & \rho_{25} & \rho_{26}\\
{\color{red}\rho_{31}} & {\color{red}\rho_{32}} & {\color{red}\rho_{33}} & \rho_{34} & \rho_{35} & \rho_{36}\\
\rho_{41} & \rho_{42} & \rho_{43} & {\color{green}\rho_{44}} & {\color{green}\rho_{45}} & \rho_{46}\\
\rho_{51} & \rho_{52} & \rho_{53} & {\color{green}\rho_{54}} & {\color{green}\rho_{55}} & \rho_{56}\\
\rho_{61} & \rho_{62} & \rho_{63} & \rho_{64} & \rho_{65} & {\color{blue}\rho_{66}}
\end{pmatrix}\\
\\
\downarrow tr_{\left(A\right)}\\
\\
\begin{pmatrix}{\color{red}\rho_{11}}+{\color{green}\rho_{44}}+{\color{blue}\rho_{66}} & {\color{red}\rho_{12}}+{\color{green}\rho_{45}} & {\color{red}\rho_{13}}\\
{\color{red}\rho_{21}}+{\color{green}\rho_{54}} & {\color{red}\rho_{22}}+{\color{green}\rho_{55}} & {\color{red}\rho_{23}}\\
{\color{red}\rho_{31}} & {\color{red}\rho_{32}} & {\color{red}\rho_{33}}
\end{pmatrix}
\end{array}
\]
The colored blocks (color online) of the top matrix correspond to
the elements $\left|\psi_{i}\right\rangle \left\langle \psi_{i}\right|$.
From this we learn how to ``read'' the action of QCG from the bipartition
table: 
\begin{enumerate}
\item Coherences between basis elements $\left|\gamma_{i,k}\right\rangle \left\langle \gamma_{j,l}\right|$
in different rows ($i\neq j$) of the bipartition table are discarded. 
\item For each pair of columns $k,l$ (including $k=l$), the sum of coherences
between $\left|\gamma_{i,k}\right\rangle \left\langle \gamma_{i,l}\right|$
over all rows $i$, is the new coherence term for the reduced element
$\left|\beta_{k}\right\rangle \left\langle \beta_{l}\right|$.
\end{enumerate}
The original Hilbert space can then be decomposed to sectors
\begin{equation}
\mathcal{H}=\bigoplus_{k=1}^{m}\mathcal{H}_{k},\label{eq: H=00003D bigOsum H_k}
\end{equation}
where $\mathcal{H}_{k}$ is the span of elements in column $k$ of
the bipartition table. This decomposition is analogous to the partition
of the classical state-space to blocks. The rule 2 above suggests
that a state supported on a single column $\mathcal{H}_{k}$ collapses
into a macro state $\left|\beta_{k}\right\rangle $, as in the classical
case. Similarly, all statistical mixtures of states supported on different
columns collapse into statistical mixtures of the corresponding macro
states. The quantum-classical similarities end when we consider superpositions
between the blocks. QCG attempts to reduce the coherence terms between
the blocks into a single coherence term between the corresponding
macro states, but it cannot do so perfectly. The result is potentially
diminished coherence between the macro states of the reduced state.

Although the visual representation of QCG in terms of columns and
rows of the bipartition table is appealing, its operational meaning
is not clear. In the following we will identify a set of operators
that capture the structure of the bipartition table and use them to
gain insight about QCG's operational meaning. 

\subsubsection{Bipartition Operators}

Similarly to how we obtained the operator sum representation of $tr_{\left(A\right)}$,
we can obtain another representation by using the first form of $V$
in Eq. (\ref{eq:def of V})

\begin{align*}
tr_{\left(A\right)}\left(\rho\right) & =tr_{A}\left(V\rho V^{\dagger}\right)\\
 & =tr_{A}\left[\left(\sum_{l=1}^{m}V_{l}^{A}\otimes\left|\beta_{l}\right\rangle \right)\rho\left(\sum_{k=1}^{m}V_{k}^{A\dagger}\otimes\bra{\beta_{k}}\right)\right]\\
 & =\sum_{k,l=1}^{m}tr\left[V_{k}^{A\dagger}V_{l}^{A}\rho\right]\left|\beta_{l}\right\rangle \bra{\beta_{k}}.
\end{align*}
This brings us to the definition of the bipartition operators 
\begin{align}
S_{kl} & :=V_{k}^{A\dagger}V_{l}^{A}=\sum_{i=1}^{\mathsf{\mathsf{min}}\left(h_{k},h_{l}\right)}\left|\gamma_{i,k}\right\rangle \left\langle \gamma_{i,l}\right|\label{eq:def of bipartition operators}
\end{align}
that map between columns of the bipartition table by preserving the
row index $i$ of each element (the element is eliminated if the row
is not present in the destination column). 

As a result, we obtain a different representation of the QCG map (in
\cite{Milz17} such representation of quantum channels is described
as input/output or tomographic representation)
\begin{align}
tr_{\left(A\right)}\left(\rho\right) & =\sum_{k,l}tr\left(S_{kl}\rho\right)\left|\beta_{l}\right\rangle \left\langle \beta_{k}\right|.\label{eq:tr_(A)  action with S_kl}
\end{align}
Since the bipartition operators can be read directly from the bipartition
table, from now on we will use the right hand side of Eq. (\ref{eq:def of bipartition operators})
and Eq. (\ref{eq:tr_(A)  action with S_kl}) as the defining constructs
of QCG and leave the isometry $V$ behind (the bipartition table is
of course still the underlying structure from which all of these constructs
are derived). 

In order to obtain the operational meaning of QCG, consider what observable
information is preserved in the reduced state. Formally, the information
in the reduced state $\rho_{B}:=tr_{\left(A\right)}\left(\rho\right)$
predicts, according to Born's rule, the expectation values $tr\left(O_{B}\rho_{B}\right)$
for all observables $O_{B}$ in $\mathcal{O}\left(\mathcal{H}^{B}\right)$
(the set of observables on $B$). Since Born's rule is identical to
the Hilbert-Schmidt (HS) inner product $\left\langle O_{B},\rho_{B}\right\rangle _{HS}:=tr\left(O_{B}^{\dagger}\rho_{B}\right)$,
we can lift the QCG map from states and apply it to observables 
\[
tr\left(O_{B}\rho_{B}\right)=\left\langle O_{B},tr_{\left(A\right)}\left(\rho\right)\right\rangle _{HS}=\left\langle tr_{\left(A\right)}^{\dagger}\left(O_{B}\right),\rho\right\rangle _{HS}
\]
where $tr_{\left(A\right)}^{\dagger}$ is the Hermitian adjoint of
$tr_{\left(A\right)}$ with respect to the HS inner product (the same
symbol $^{\dagger}$ for Hermitian adjoint will be used for both operators
and superoperators). The following set of observables on the original
(unreduced) system
\begin{equation}
\mathcal{O}^{B}\left(\mathcal{H}\right):=\left\{ tr_{\left(A\right)}^{\dagger}\left(O_{B}\right)\,|\,O_{B}\in\mathcal{O}\left(\mathcal{H}^{B}\right)\right\} \subset\mathcal{O}\left(\mathcal{H}\right),\label{eq: def of rest obs set}
\end{equation}
consists of all the observables whose expectation values are preserved
by QCG. 

The explicit form of $tr_{\left(A\right)}^{\dagger}$ can be derived
by rearranging the traces and sums in 
\begin{align*}
\left\langle O_{B},tr_{\left(A\right)}\left(O\right)\right\rangle _{HS} & =tr\left(O_{B}^{\dagger}\sum_{k,l}tr\left(S_{kl}O\right)\left|\beta_{l}\right\rangle \left\langle \beta_{k}\right|\right)\\
 & =tr\left(\sum_{k,l}S_{kl}tr\left(O_{B}^{\dagger}\left|\beta_{l}\right\rangle \left\langle \beta_{k}\right|\right)O\right)\\
 & =\left\langle \left(\sum_{k,l}S_{kl}\left\langle \beta_{k}\right|O_{B}^{\dagger}\left|\beta_{l}\right\rangle \right)^{\dagger},O\right\rangle _{HS}.
\end{align*}
Then, using $S_{kl}=S_{lk}^{\dagger}$ and rearranging the indices
we get 
\begin{equation}
tr_{\left(A\right)}^{\dagger}\left(O_{B}\right)=\sum_{k,l}S_{kl}\left\langle \beta_{k}\right|O_{B}\left|\beta_{l}\right\rangle .\label{eq:def of tr_dagger_(A)}
\end{equation}

It is now clear that $\mathcal{O}^{B}\left(\mathcal{H}\right)\subset\textrm{\ensuremath{\mathsf{span}}}\left\{ S_{kl}\right\} $.
Conversely, for every observable $O\in\textrm{\ensuremath{\mathsf{span}}}\left\{ S_{kl}\right\} $
we can find an $O_{B}\in\mathcal{O}\left(\mathcal{H}^{B}\right)$
s.t. $O=tr_{\left(A\right)}^{\dagger}\left(O_{B}\right)$. Therefore,
bipartition operators $S_{kl}$ span the operator subspace containing
all and only the observables preserved by QCG. Then we can interpret
the coarse-grained state $\rho_{B}$ as the state that contains all
and only the information that is accessible to observer restricted
to $\textrm{\ensuremath{\mathsf{span}}}\left\{ S_{kl}\right\} $.
QCG map can then be understood as a change-of-observer transformation.

In the familiar case of tensor product bipartition $\mathcal{H}=\mathcal{H}^{A}\otimes\mathcal{H}^{B}$,
bipartition operators take the form 
\[
S_{kl}:=I_{A}\otimes\ket{\beta_{k}}\bra{\beta_{l}}.
\]
The restricted set of observables $\textrm{\ensuremath{\mathsf{span}}}\left\{ S_{kl}\right\} =I_{A}\otimes\mathcal{B}\left(\mathcal{H}^{B}\right)$
imply that the observer can only measure system $B$. The QCG map
(\ref{eq:tr_(A)  action with S_kl}) specializes to the usual $tr_{A}$
and the reduced states $tr_{A}\left(\rho\right)$ represent what the
restricted observer can actually ``see''. In Section \ref{subsec:Examples:-Operationally-Motivate}
we will see other familiar state transformations that can be understood
as special cases of QCG.

This closes the circle with the classical picture of CG from which
we started. Classical CG was introduced as the manifestation observer's
inability to distinguish some states, which is in fact a restriction
of observational power. Now we see that both classical and quantum
notions admit the same operational interpretation: CG is the result
of restricted observational ability. 

\subsubsection{Generalization: Quantum-Classical Hybrid\label{subsec:Hybrid-Coarse-Grainings}}

With bipartition operators it is easy to extend the quantum notion
of CG to include the original classical one. Intermediate notions,
that combine both classical and quantum features, are quick to follow
(we will keep referring to them as QCG). This generalization will
allow us to associate QCG with symmetries in section \ref{subsec:Coarse Graining and Symmetries}.

The purely classical notion of CG can be imported into quantum state-space
by simply disregarding the coherence terms. Using the set $\left\{ \varPi_{k}\right\} $
of projections on sectors $\mathcal{H}=\bigoplus_{k=1}^{m}\mathcal{H}_{k}$
that specify the classical blocks, the classical CG map is defined
as
\begin{equation}
\rho\longmapsto\sum_{k}tr\left(\varPi_{k}\rho\right)\ketbra{\beta_{k}}{\beta_{k}}.\label{eq:classical CG of density matrix}
\end{equation}
One can always represent probability vectors as diagonal density matrices
and use this map to implement classical CG as defined by Eq. (\ref{eq: p_B(b)=00003DSum_a_(p_S(a))}).
Comparing Eq. (\ref{eq:classical CG of density matrix}) to the quantum
version (\ref{eq:tr_(A)  action with S_kl}) suggests that the set
of projections $\left\{ \varPi_{k}\right\} $ is the classical equivalent
of the bipartition operators. In fact, note that by definition (\ref{eq:def of bipartition operators}),
bipartition operators of the form $S_{kk}$ are projections on sectors.
If we think of bipartition operators as $k,l$ elements of some matrix,
then $S_{kk}$ are the diagonal elements. Classical CG can then be
thought of as a restriction of some QCG specified by $\left\{ S_{kl}\right\} $
to the diagonal elements $\left\{ S_{kk}\right\} $. 

This perspective leaves room for intermediate cases that arise from
restriction of the complete set $\left\{ S_{kl}\right\} $ to block
diagonal elements. It is convenient to introduce the index $q$ to
refer to the blocks of bipartition operators, s.t. $\left\{ S_{q,kl}\right\} _{kl}$
is a block diagonal set with $k,l$ running over the elements of block
$q$. The hybrid CG map is then specified by the set $\left\{ S_{q,kl}\right\} $
and it acts similarly to (\ref{eq:tr_(A)  action with S_kl}),
\begin{align}
\rho\longmapsto & \sum_{q,k,l}tr\left(S_{q,kl}\rho\right)\left|\beta_{q,l}\right\rangle \left\langle \beta_{q,k}\right|,\label{eq:Hybrid CG map}
\end{align}
with the addition of index $q$. The purely quantum case is when $q$
specifies a single block, making the index $q$ unnecessary. The purely
classical case is when each block $q$ has only one bipartition operator
\textendash{} the projection $\varPi_{q}$. The truly hybrid case
selects the super-sectors $\mathcal{H}=\bigoplus_{q}\mathcal{H}_{q}$
of the Hilbert space where each subset $\left\{ S_{q,kl}\right\} _{kl}$
of bipartition operators is supported. The map (\ref{eq:Hybrid CG map})
reduces each super-sector $\mathcal{H}_{q}$ into a distinct sector
in the reduced state space while discarding all coherence terms between
the different $\mathcal{H}_{q}$. 

We can also generalize the visual representation of QCG with bipartition
tables by allowing block diagonal arrangements of cells. For each
subset of operators $\left\{ S_{q,kl}\right\} _{kl}$ we have a block
of cells in the bipartition table, and the different blocks live on
the diagonal of the full table
\begin{center}
\begin{tabular}{ccc}
\begin{tabular}{|c|c|c}
\hline 
$\gamma_{1;1,1}$ & $\gamma_{1;1,2}$ & \multicolumn{1}{c|}{$...$}\tabularnewline
\hline 
$\gamma_{1;2,1}$ & $\ddots$ & \tabularnewline
\cline{1-2} 
$\vdots$ & \multicolumn{1}{c}{} & \tabularnewline
\cline{1-1} 
\end{tabular} &  & \tabularnewline
 & %
\begin{tabular}{|c|c|c}
\hline 
$\gamma_{2;1,1}$ & $\gamma_{2;1,2}$ & \multicolumn{1}{c|}{$...$}\tabularnewline
\hline 
$\gamma_{2;2,1}$ & $\ddots$ & \tabularnewline
\cline{1-2} 
$\vdots$ & \multicolumn{1}{c}{} & \tabularnewline
\cline{1-1} 
\end{tabular} & \tabularnewline
 &  & $\ddots$\tabularnewline
\end{tabular}
\par\end{center}

\noindent This arrangement results in the block diagonal set $\left\{ S_{q,kl}\right\} $
if we use the original construction (\ref{eq:def of bipartition operators})
of bipartition operators with such tables.

\subsection{Special Cases of the Coarse-Graining Map \label{subsec:Examples:-Operationally-Motivate}}

The general QCG map (\ref{eq:Hybrid CG map}) captures a lot of common
state manipulations \textendash{} which are not usually thought of
as CG \textendash{} as its special cases. Since the QCG map is completely
specified by the set of bipartition operators it is possible to capture
the key structure associated with such manipulations in the neat visual
form of the bipartition table. In the following we point out a few
of such state manipulations. 

For concreteness we will consider the system of two or more spin-$\frac{1}{2}$
particles
\begin{align*}
\mathcal{H}:= & \left(\mathcal{H}^{\left(\frac{1}{2}\right)}\right)^{\otimes N}\,\,\,\,\,\,\,\,\,\,\,\mathcal{H}^{\left(\frac{1}{2}\right)}=\mathsf{span\left\{ \ket{\uparrow},\ket{\downarrow}\right\} }.
\end{align*}

\paragraph{Change of basis:}

The trivial QCG that does not actually loose any information may still
change the basis in which the density matrix is presented. The change
of basis map, disguised as QCG, is specified by arranging the new
basis elements into a single row of the bipartition table. For 2 spins,
changing to the total spin basis $\ket{j,m}$ is given by the table
\begin{center}
\begin{tabular}{|c|c|c|c|}
\hline 
$1,1$ & $1,0$ & $1,-1$ & $0,0$\tabularnewline
\hline 
\end{tabular}
\par\end{center}

\noindent which specifies the bipartition operators
\[
S_{j,m;j',m'}:=\ketbra{j,m}{j',m'},
\]
where $j,m$ are used to refer to the columns of the table. The QCG
map then simply changes the basis 
\begin{align*}
\rho\longmapsto & \sum_{j,m;j',m'}tr\left(S_{j,m;j',m'}\rho\right)\ketbra{j',m'}{j,m}=\\
 & \sum_{j,m;j',m'}\bra{j',m'}\rho\ket{j,m}\ket{j',m'}\bra{j,m}.
\end{align*}
This should make clear the fact that the result of any QCG, even the
trivial one, depends on the choice of basis that go into the bipartition
table.

\paragraph{Projective Measurement:}

Projective measurements, up to the readout of the outcome, can be
thought of as purely classical CGs. Here the bipartition table has
a column-diagonal form and the columns are specified by the projections
on the outcomes. For 2 spins, the QCG resulting from measurement of
the total spin $z$ component (without reading the outcome) is specified
by the table
\begin{center}
\begin{tabular}{c|c|c}
\cline{1-1} 
\multicolumn{1}{|c|}{$\uparrow\uparrow$} & \multicolumn{1}{c}{} & \tabularnewline
\cline{1-2} 
 & $\uparrow\downarrow$ & \tabularnewline
\cline{2-2} 
 & $\downarrow\uparrow$ & \tabularnewline
\cline{2-3} 
\multicolumn{1}{c}{} &  & \multicolumn{1}{c|}{$\downarrow\downarrow$}\tabularnewline
\cline{3-3} 
\end{tabular}
\par\end{center}

\noindent There are only 3 bipartition operators defined by this table:
the projections 
\begin{align*}
S_{1,1} & =\ketbra{\uparrow\uparrow}{\uparrow\uparrow}\\
S_{0,0} & =\ketbra{\uparrow\downarrow}{\uparrow\downarrow}+\ketbra{\downarrow\uparrow}{\downarrow\uparrow}\\
S_{-1,-1} & =\ketbra{\downarrow\downarrow}{\downarrow\downarrow},
\end{align*}
where $j_{z}=1,0,-1$ are used to label the columns. The associated
QCG map 
\[
\rho\longmapsto\sum_{j_{z}=-1,0,1}tr\left(S_{j_{z},j_{z}}\rho\right)\ketbra{j_{z}}{j_{z}}
\]
results in a diagonal matrix containing the probability distribution
over the three outcomes. 

\paragraph{Tensor product structures and (virtual) subsystems: }

Illustrating bipartite tensor product structures is where the bipartition
table really simplifies the picture. The natural tensor product structure
of the Hilbert space of 2 spins $A$ and $B$ is captured by the bipartition
table
\begin{center}
\begin{tabular}{|c|c|}
\hline 
$\uparrow\uparrow$ & $\uparrow\downarrow$\tabularnewline
\hline 
$\downarrow\uparrow$ & $\downarrow\downarrow$\tabularnewline
\hline 
\end{tabular}
\par\end{center}

\noindent It is arranged such that the degrees of freedom of spin
$A$ are constant inside the rows and the degrees of freedom of spin
$B$ are constant inside the columns. This table defines the bipartition
operators $S_{kl}:=I\otimes\ketbra kl$ for $k,l=\uparrow,\downarrow$
and the associated QCG map is just the partial trace over $A$ (rotate
the table by $90^{\circ}$ to get the partial trace over $B$)
\begin{align*}
\rho\longmapsto & \sum_{k,l=\uparrow,\downarrow}tr\left(I\otimes\ketbra kl\rho\right)\ketbra lk=\\
 & \sum_{k,l=\uparrow,\downarrow}tr\left(\ketbra kltr_{A}\left(\rho\right)\right)\ketbra lk=tr_{A}\left(\rho\right).
\end{align*}

For 3 spins we can consider the first 2 spins as subsystem $A$ and
the 3rd spin as subsystem $B$. Arranging the bipartition table where
$B$'s degrees of freedom are constant inside the columns and $A$'s
inside the rows results in
\begin{center}
\begin{tabular}{|c|c|}
\hline 
$\uparrow\uparrow\uparrow$ & $\uparrow\uparrow\downarrow$\tabularnewline
\hline 
$\uparrow\downarrow\uparrow$ & $\uparrow\downarrow\downarrow$\tabularnewline
\hline 
$\downarrow\uparrow\uparrow$ & $\downarrow\uparrow\downarrow$\tabularnewline
\hline 
$\downarrow\downarrow\uparrow$ & $\downarrow\downarrow\downarrow$\tabularnewline
\hline 
\end{tabular}
\par\end{center}

\noindent which specifies a QCG map that traces out the first 2 spins.
By rearranging this table we can specify different (possibly virtual)
bipartite tensor product structures. 

For example 
\begin{center}
\begin{tabular}{|c|c|}
\hline 
$\uparrow\uparrow\uparrow$ & $\downarrow\downarrow\downarrow$\tabularnewline
\hline 
$\downarrow\uparrow\uparrow$ & $\uparrow\downarrow\downarrow$\tabularnewline
\hline 
$\uparrow\downarrow\uparrow$ & $\downarrow\uparrow\downarrow$\tabularnewline
\hline 
$\uparrow\uparrow\downarrow$ & $\downarrow\downarrow\uparrow$\tabularnewline
\hline 
\end{tabular}
\par\end{center}

\noindent specifies the natural tensor product structure of the repetition
code. The virtual subsystem associated with the columns now encodes
the logical qubit, while the virtual subsystem associated with the
rows encodes the syndrome. The 4 bipartition operators consist of
2 projections $S_{00}$, $S_{11}$ on the columns ($0$,$1$ label
the two columns), and 2 isometries $S_{01}$, $S_{10}$ between the
columns that exchange elements inside the rows (recall Eq. (\ref{eq:def of bipartition operators})
for explicit definition). A single spin flip error $X_{i}$ acts on
the top row \textendash{} the code space \textendash{} by translating
it to the $i+1$ row, so
\[
S_{kl}X_{i}\ket{\uparrow\uparrow\uparrow}=X_{i}S_{kl}\ket{\uparrow\uparrow\uparrow}
\]
\[
S_{kl}X_{i}\ket{\downarrow\downarrow\downarrow}=X_{i}S_{kl}\ket{\downarrow\downarrow\downarrow}.
\]
Therefore, for any encoding $\ket{\psi}=\alpha\ket{\uparrow\uparrow\uparrow}+\beta\ket{\downarrow\downarrow\downarrow}$
we can have a single spin flip error that will not affect the coarse
grained state
\begin{align*}
X_{i}\ketbra{\psi}{\psi}X_{i}\longmapsto & \sum_{k,l=0,1}tr\left(S_{kl}X_{i}\ketbra{\psi}{\psi}X_{i}\right)\ketbra lk=\\
 & \sum_{k,l=0,1}tr\left(S_{kl}\ketbra{\psi}{\psi}\right)\ketbra lk=\\
 & \left(\alpha\ket 0+\beta\ket 1\right)\left(\overline{\alpha}\bra 0+\overline{\beta}\bra 1\right).
\end{align*}

\noindent In this context we think of the QCG map as a decoding procedure
that traces out the syndrome degrees of freedom and produces the encoded
qubit. 

\paragraph{Reference frames and noiseless subsystems:}

For errors that arbitrarily change the reference frame (RF) there
are \textit{noiseless subsystems} where information can be encoded
in RF-independent degrees of freedom \cite{Lidar13}. Such degrees
of freedom can be associated with the reduced state that is seen by
an observer that does not have access to the RF in which the state
was prepared \cite{Bartlet07}. This reduction of state can also be
considered as QCG. 

Since RFs are completely specified by a group of transformations that
change them, the relevant structure of QCG is selected by the irreducible
representations of the group (we will elaborate on this in Section
\ref{subsec:Coarse Graining and Symmetries}). Considering a system
of three spins and a RF of direction associated with global rotations,
we get the bipartition table 
\begin{center}
\begin{tabular}{|c|cc}
\cline{1-1} 
$\frac{3}{2},+\frac{3}{2}$ & {\large{}\strut } & \tabularnewline
\cline{1-1} 
$\vdots$ & {\large{}\strut } & \tabularnewline
\cline{1-1} 
$\frac{3}{2},-\frac{3}{2}$ & {\large{}\strut } & \tabularnewline
\hline 
\multicolumn{1}{c|}{{\large{}\strut }} & \multicolumn{1}{c|}{$\frac{1}{2},+\frac{1}{2},0$} & \multicolumn{1}{c|}{$\frac{1}{2},+\frac{1}{2},1$}\tabularnewline
\cline{2-3} 
\multicolumn{1}{c|}{{\large{}\strut }} & \multicolumn{1}{c|}{$\frac{1}{2},-\frac{1}{2},0$} & \multicolumn{1}{c|}{$\frac{1}{2},-\frac{1}{2},1$}\tabularnewline
\cline{2-3} 
\end{tabular}
\par\end{center}

\noindent There are two blocks in this table corresponding to the
irreducible representations of total spin $\frac{3}{2}$ and $\frac{1}{2}$.
The block of total spin $\frac{1}{2}$ specifies the virtual tensor
product $\ket{\frac{1}{2},\pm\frac{1}{2}}\otimes\ket k$ where $k=0,1$
labels the two copies of this representation. The 5 bipartition operators
specified by this table are
\begin{align*}
S_{\frac{3}{2}} & :=\sum_{j_{z}=-\frac{3}{2},...,\frac{3}{2}}\ketbra{\frac{3}{2},j_{z}}{\frac{3}{2},j_{z}}=I^{\left(\frac{3}{2}\right)}\\
\\
S_{\frac{1}{2},kl} & :=\sum_{j_{z}=-\frac{1}{2},\frac{1}{2}}\ketbra{\frac{1}{2},j_{z},k}{\frac{1}{2},j_{z},l}=I^{\left(\frac{1}{2}\right)}\otimes\ketbra kl
\end{align*}

\noindent and, according to Schur's lemmas \cite{Cornwell}, they
span the space of all operators that commute with all global rotations
that act on $\left(\mathcal{H}^{\left(\frac{1}{2}\right)}\right)^{\otimes3}$.
Therefore, according to the operational interpretation of QCG, the
reduced state retains only the information accessible with rotationally
invariant measurements. Such restriction of measurements is what defines
the observer that has no access to the RF of direction \cite{Bartlet07}
so this QCG produces the effective state that such observer can see. 

In the context of noiseless subsystems we can say that such QCG ``traces
out'' rotationally non-invariant degrees of freedom and produces
the qubit encoded in the rotationally invariant degrees of freedom. 

\subsection{Compatibility with Dynamics\label{subsec:quantum Compatibility-with-Dynamics}}

As was discussed in the classical case, the coarse-grained state may
fail to follow a well defined dynamical rule. The dynamics in the
coarse-grained state-space may be such that it is impossible to tell,
from the initial conditions alone, where the system will go. The situation
is essentially the same as the one we see in open quantum systems
(see \cite{Breuer02} or \cite{Rivas12} for a comprehensive review).
In fact it was recently shown \cite{Duarte17} that under dimension
reducing maps, such as our QCG map, the reduced dynamics can be described
in the same way we describe the dynamics of open quantum systems.
This conclusion should also be evident from the approach to QCG we
have developed here: if QCG is a marginalization of a (partial) subsystem
then the remaining subsystem should evolve as an open quantum system.
This means that in general the evolution may not be universal, so
the dynamical map that governs the evolution is different for different
initial conditions and may not be completely positive \cite{Rivas12}.
Even when the dynamics are universal we may still loose the semigroup
structure which allows to characterize the dynamics with generators. 

How to deal with these difficulties in the context of open quantum
systems is an area of active research \cite{Breuer16} and we will
not attempt to address it here. Our situation is different in that
we have the freedom to choose the bipartition that may be compatible
with the given dynamics. Instead of asking how a fixed subsystem evolves,
we ask how to choose a (partial) subsystem so it evolves in a nice
way. This question will be now addressed in the form of compatibility
condition between QCGs and dynamics. 

The way the compatibility condition was derived in the classical case
(section \ref{subsec:classical Compatibility-with-Dynamics}) is sufficiently
general to be reproduced in the quantum setting. The classical condition
$PQ=PQP$ of Theorem \ref{thm:PQ=00003DPQP} has two components \textendash{}
the generator of dynamics $Q$, and the CG projection $P$. Since
the QCG map $tr_{\left(A\right)}$ is a superoperator that acts on
density matrices, the quantum analogues of $Q$ and $P$ must also
be superoperators. The analogue of $Q$ is the Lindblad superoperator
$\mathcal{L}$ \cite{Rivas12,Breuer02}, that generates time evolutions
of the density matrix $\rho$ as the solutions of 
\begin{equation}
\frac{d}{dt}\rho=\mathcal{L}\left(\rho\right).\label{eq: Markovian dynamics diff-eq for a quantum state}
\end{equation}
This equation is the quantum analogue of Eq. (\ref{eq:Markovian dynamics diff-eq for a classical state}). 

The QCG projection can be defined identically to its classical analogue
$P=M^{+}M$ as

\begin{equation}
\mathcal{P}:=tr_{\left(A\right)}^{+}\circ tr_{\left(A\right)},\label{eq:P:=00003Dtr_A^+  tr_A}
\end{equation}
where $tr_{\left(A\right)}^{+}$ is the Moore-Penrose pseudo inverse
of the QCG map $tr_{\left(A\right)}$. Explicit form of this pseudo
inverse is not necessary for our purposes and we will only use its
defining properties and the fact that it exists (all finite dimensional
linear operators \textendash{} including $tr_{\left(A\right)}$ \textendash{}
have one). For the sake of completeness we will present the explicit
forms of $tr_{\left(A\right)}^{+}$ and $\mathcal{P}$ after proving
Lemma \ref{lem:P projects on span=00007BS_kl=00007D} bellow. 

With these definitions Theorem \ref{thm:PQ=00003DPQP} can be reproduced
in the quantum setting by replacing $P$ with $\mathcal{P}$, $Q$
with $\mathcal{L}$, $M$ ($M^{+}$) with $tr_{\left(A\right)}$ ($tr_{\left(A\right)}^{+}$),
probability vectors with density matrices, and Eq. (\ref{eq:Markovian dynamics diff-eq for a classical state})
with Eq. (\ref{eq: Markovian dynamics diff-eq for a quantum state}).
The proof is the same because it relies on the linear algebraic properties
of the operators (which in both cases are assumed to be finite dimensional
matrices) and nothing more. The result is that
\begin{equation}
\mathcal{P}\mathcal{L}=\mathcal{P}\mathcal{L}\mathcal{P}\label{eq:PL=00003DPLP}
\end{equation}
is the quantum compatibility condition in the general form. When this
condition holds, and only then, the reduced state $\rho_{B}:=tr_{\left(A\right)}\left(\rho\right)$
evolves according to
\[
\frac{d}{dt}\rho_{B}=\tilde{\mathcal{L}}\left(\rho_{B}\right),
\]
where $\tilde{\mathcal{L}}=tr_{\left(A\right)}\circ\mathcal{L}\circ tr_{\left(A\right)}^{+}$.

The compatibility condition in its general form (\ref{eq:PL=00003DPLP})
is quite opaque. In the classical case it was corollary \ref{cor:total rates must be uniform}
that provided some insight into how to find compatible CG by looking
at transition rates. Extracting similar insight for quantum dynamics
is not as easy. We will not address this general case here but we
will specialize the generator $\mathcal{L}$ to unitary dynamics (given
by a Hamiltonian) and reformulate the condition (\ref{eq:PL=00003DPLP})
in a more transparent way. In the next subsection we will specialize
this condition further by focusing on QCGs given by a group representation.

The first step in clarifying the condition (\ref{eq:PL=00003DPLP})
is finding out the operator subspace on which $\mathcal{P}$ projects. 
\begin{lem}
\label{lem:P projects on span=00007BS_kl=00007D}Let $\left\{ S_{kl}\right\} $
be a set of bipartition operators, and let $\mathcal{P}$ be the associated
coarse-graining projection as defined by Eq. (\ref{eq:P:=00003Dtr_A^+  tr_A}).
Then $\mathcal{P}$ is an orthogonal projection on the operator subspace
$\textrm{\ensuremath{\mathsf{span}}}\left\{ S_{kl}\right\} $. 
\end{lem}
\begin{proof}
The defining properties of the Moore-Penrose pseudo inverse \cite{Penrose54}
imply that the map $\mathcal{P}$ is an orthogonal projection on the
subspace orthogonal to the kernel of $tr_{\left(A\right)}$, that
is $\textsf{im}\left(\mathcal{P}\right)=\mathsf{ker}\left(tr_{\left(A\right)}\right)^{\perp}$.
Next, to see that $\mathsf{ker}\left(tr_{\left(A\right)}\right)^{\perp}=\textrm{\ensuremath{\mathsf{span}}}\left\{ S_{kl}\right\} $,
we will apply $tr_{\left(A\right)}$ on $S_{kl}$. Using the action
(\ref{eq:Tr_(a) acting on matrix element}) on the definition (\ref{eq:def of bipartition operators})
we get
\begin{align}
tr_{\left(A\right)}\left(S_{kl}\right) & =\mathsf{\mathsf{min}}\left(h_{k},h_{l}\right)\left|\beta_{k}\right\rangle \left\langle \beta_{l}\right|.\label{eq:tr_(A)(S_kl)}
\end{align}
From this we see that the image of $\textrm{\ensuremath{\mathsf{span}}}\left\{ S_{kl}\right\} $
under $tr_{\left(A\right)}$ is the whole $\mathsf{im}\left(tr_{\left(A\right)}\right)$.
The minimal subspace with such property is $\mathsf{ker}\left(tr_{\left(A\right)}\right)^{\perp}$,
therefore $\mathsf{ker}\left(tr_{\left(A\right)}\right)^{\perp}\subseteq\textrm{\ensuremath{\mathsf{span}}}\left\{ S_{kl}\right\} $.
On the other hand, every non-zero operator in $\textrm{\ensuremath{\mathsf{span}}}\left\{ S_{kl}\right\} $
does not vanish under $tr_{\left(A\right)}$, therefore $\textrm{\ensuremath{\mathsf{span}}}\left\{ S_{kl}\right\} \subseteq\mathsf{ker}\left(tr_{\left(A\right)}\right)^{\perp}$.
The two mutual inclusions then imply 
\[
\textrm{\ensuremath{\mathsf{span}}}\left\{ S_{kl}\right\} =\mathsf{ker}\left(tr_{\left(A\right)}\right)^{\perp}=\textsf{im}\left(\mathcal{P}\right).
\]
\end{proof}
Now we note that the pseudo inverse $tr_{\left(A\right)}^{+}$ is
a map from $\textsf{im}\left(tr_{\left(A\right)}\right)$ to $\mathsf{ker}\left(tr_{\left(A\right)}\right)^{\perp}$,
that is
\[
tr_{\left(A\right)}^{+}:\textsf{span}\left\{ \ketbra{\beta_{k}}{\beta_{l}}\right\} \rightarrow\textrm{\ensuremath{\mathsf{span}}}\left\{ S_{kl}\right\} .
\]
Eq. (\ref{eq:tr_(A)(S_kl)}) suggests that for the inverse property
$tr_{\left(A\right)}\circ tr_{\left(A\right)}^{+}=\mathcal{I}$ to
hold we must have $tr_{\left(A\right)}^{+}\left(\ketbra{\beta_{k}}{\beta_{l}}\right)=\mathsf{\mathsf{min}}\left(h_{k},h_{l}\right)^{-1}S_{kl}$
which defines the pseudo inverse 
\[
tr_{\left(A\right)}^{+}\left(O_{B}\right)=\sum_{kl}\frac{\bra{\beta_{k}}O_{B}\ket{\beta_{l}}}{\mathsf{\mathsf{min}}\left(h_{k},h_{l}\right)}S_{kl}.
\]
This map can be seen as a composition of $tr_{\left(A\right)}^{\dagger}$
(see Eq. (\ref{eq:def of tr_dagger_(A)})) with rescaling by $\mathsf{\mathsf{min}}\left(h_{k},h_{l}\right)$. 

The explicit form of $\mathcal{P}=tr_{\left(A\right)}^{+}\circ tr_{\left(A\right)}$
is then given by acting with $tr_{\left(A\right)}^{+}$ on Eq. (\ref{eq:tr_(A)  action with S_kl})
\begin{align*}
\mathcal{P}\left(O\right) & =\sum_{k,l}tr\left(S_{kl}O\right)tr_{\left(A\right)}^{+}\left(\left|\beta_{l}\right\rangle \left\langle \beta_{k}\right|\right)\\
 & =\sum_{k,l}\frac{tr\left(S_{kl}O\right)}{\mathsf{\mathsf{min}}\left(h_{k},h_{l}\right)}S_{lk}.
\end{align*}

It should be noted that even though the QCG $tr_{\left(A\right)}$
maps states to states (is CPTP), we cannot claim that $tr_{\left(A\right)}^{+}$
and $\mathcal{P}$ have this property in general. Nonetheless, the
QCG projection $\mathcal{P}$ is a useful formal construct that captures
the compatibility condition (\ref{eq:PL=00003DPLP}) and its properties
will be used in the proof of Theorem \ref{thm:=00005BH,S=00005D in span=00007BS_kl=00007D}. 

In the following we will use the fact that $\mathcal{P}$ is an\textit{
orthogonal} projection, as stated by Lemma \ref{lem:P projects on span=00007BS_kl=00007D},
meaning that not only $\mathcal{P}^{2}=\mathcal{P}$ but also $\mathcal{P}^{\dagger}=\mathcal{P}$
(the Hermitian adjoint is defined with respect to the HS inner product
$\left\langle \mathcal{P}\left(A\right),B\right\rangle _{HS}=\left\langle A,\mathcal{P}^{\dagger}\left(B\right)\right\rangle _{HS}$). 

Now we will assume unitary dynamics. This means that the generator
$\mathcal{L}$ is of the form $-i\left[H,\cdot\right]$, where $H$
is the Hamiltonian. The following theorem expresses the compatibility
condition (\ref{eq:PL=00003DPLP}) in terms of $H$ and $\left\{ S_{kl}\right\} $. 
\begin{thm}
\label{thm:=00005BH,S=00005D in span=00007BS_kl=00007D} Let $\mathcal{L}\left(\cdot\right):=-i\left[H,\cdot\right]$
be a generator of dynamics with Hamiltonian $H$, and let $\left\{ S_{kl}\right\} $
be bipartition operators that specify a coarse-graining. Then, the
compatibility condition (\ref{eq:PL=00003DPLP}) is equivalent to
\[
\left[H,S\right]\in\textrm{\ensuremath{\mathsf{span}}}\left\{ S_{kl}\right\} \,\,\,\,\,\forall S\in\textrm{\ensuremath{\mathsf{span}}}\left\{ S_{kl}\right\} 
\]
\end{thm}
\begin{proof}
First we note that $\mathcal{L}$ is an anti-Hermitian superoperator:
$\mathcal{L}^{\dagger}=-\text{\ensuremath{\mathcal{L}}}$. This can
be shown explicitly
\begin{align*}
\left\langle A,\mathcal{L}\left(B\right)\right\rangle _{HS} & =tr\left(A^{\dagger}\left(-i\left[H,B\right]\right)\right)\\
 & =tr\left(-iA^{\dagger}HB\right)+tr\left(iA^{\dagger}BH\right)\\
 & =tr\left(-iA^{\dagger}HB\right)+tr\left(iHA^{\dagger}B\right)\\
 & =tr\left(-\mathcal{L}\left(A\right)^{\dagger}B\right)=\left\langle -\mathcal{L}\left(A\right),B\right\rangle _{HS}.
\end{align*}
By taking Hermitian adjoint on both sides of (\ref{eq:PL=00003DPLP}),
and using the fact that $\mathcal{P}^{\dagger}=\mathcal{P}$, we get
\[
-\mathcal{L}\mathcal{P}=-\mathcal{P}\mathcal{L}\mathcal{P}=-\mathcal{P}\mathcal{L}
\]
The compatibility condition is then equivalent to 
\[
\mathcal{L}\mathcal{P}=\mathcal{P}\mathcal{L}.
\]
Lemma \ref{lem:P projects on span=00007BS_kl=00007D} implies that
for any $S\in\textrm{\ensuremath{\mathsf{span}}}\left\{ S_{kl}\right\} $
we have $\mathcal{P}\left(S\right)=S$ and $\mathcal{P}\left(O\right)\in\textrm{\ensuremath{\mathsf{span}}}\left\{ S_{kl}\right\} $
for any $O$. Therefore,

\[
\left[H,S\right]=i\mathcal{L}\left(S\right)=i\mathcal{L}\mathcal{P}\left(S\right)=i\mathcal{P}\mathcal{L}\left(S\right)\in\textrm{\ensuremath{\mathsf{span}}}\left\{ S_{kl}\right\} .
\]
For the opposite direction we assume that $i\mathcal{L}\left(S\right)=\left[H,S\right]\in\textrm{\ensuremath{\mathsf{span}}}\left\{ S_{kl}\right\} $
for any $S\in\textrm{\ensuremath{\mathsf{span}}}\left\{ S_{kl}\right\} $.
Since $\mathcal{P}$ is an orthogonal projection on $\textrm{\ensuremath{\mathsf{span}}}\left\{ S_{kl}\right\} $,
for any $O$ we have $\mathcal{L}\mathcal{P}\left(O\right)\in\textrm{\ensuremath{\mathsf{span}}}\left\{ S_{kl}\right\} $,
which implies $\mathcal{L}\mathcal{P}=\mathcal{P}\mathcal{L}\mathcal{P}$. 
\end{proof}

\subsection{Coarse-Graining and Symmetries \label{subsec:Coarse Graining and Symmetries}}

As was discussed in the classical case, symmetrizing the states can
also be considered as CG. We will now reproduce this argument for
the quantum case and utilize it to address the question of reducibility
of dynamics. 

Our construction relies on structures selected by irreducible representations
(irrep) of the group and the associated operator algebras. Developments
in fault-tolerant quantum computation \cite{Knill00,Zanardi00,Kempe01},\textit{
}the study of quantum reference frames and the emergence of superselection
rules \cite{Bartlet07,Kitaev04}, and more recently, quantification
of the notion of asymmetry \cite{Marvian13,Marvian14}, have all contributed
to the establishment of the algebraic framework that we will use here.

We begin by recalling Hilbert space decompositions induced by representations
of groups \cite{Cornwell}. Given a finite or a compact Lie group
$G$, with the unitary representation $U\left(G\right)$ on the Hilbert
space $\mathcal{H}$, there is a decomposition
\begin{equation}
\mathcal{H}=\bigoplus_{q,n}\mathcal{M}_{q,n}\cong\bigoplus_{q}\mathcal{M}_{q}\otimes\mathcal{N}_{q}.\label{eq:H=00003DbigOsum_q(M_q oTimes N_q)}
\end{equation}
The sectors $\mathcal{M}_{q,n}$ carry irreps of the group, the index
$q$ runs over the inequivalent irreps, and $n$ labels the different
occurrences of the same irrep. The isomorphism on the right follows
by ``collecting'' all the equivalent irreps into a tensor product
of the virtual subsystems $\mathcal{M}_{q}$ (the irrep space), and
$\mathcal{N}_{q}$ (the multiplicity space). Then, the group action
can be expressed in the form
\begin{equation}
U\left(g\right)=\bigoplus_{q}U_{\mathcal{M}_{q}}\left(g\right)\otimes I_{\mathcal{N}_{q}}\label{eq:U(g)=00003D bigOsum( U_q otimes I)}
\end{equation}
where $U_{\mathcal{M}_{q}}\left(g\right)$ are irreducible unitary
representations of the group action. This explicitly shows that the
group acts by transforming all $\mathcal{M}_{q}$ independently according
to the irrep $q$, while leaving all $\mathcal{N}_{q}$ unaffected.

The structure (\ref{eq:H=00003DbigOsum_q(M_q oTimes N_q)}) selected
by the group (from here on, by ``group'' we refer to the group of
unitary operators acting on the Hilbert space, not the abstract representationless
group) can now be used to implement a QCG. For an isolated sector
$q$, tracing over the virtual subsystem $\mathcal{M}_{q}$ can be
seen as a QCG given by the bipartition operators 
\begin{equation}
S_{q,kl}:=I_{\mathcal{M}_{q}}\otimes\ketbra{q,k}{q,l},\label{eq:S_q,kl  for symmetries}
\end{equation}
where $\ket{q,k}$ are some arbitrary orthonormal basis of $\mathcal{N}_{q}$.
The combined set (for all $q$) of bipartition operators $\left\{ S_{q,kl}\right\} $
specifies a hybrid notion of QCG as defined by Eq. (\ref{eq:Hybrid CG map}). 

Such QCG will be called \textit{coarse-graining by symmetrization}
because it eliminates all information in the asymmetric degrees of
freedom. In order to see that explicitly, consider the commutant algebra
of the group, defined by

\[
U\left(G\right)':=\left\{ B\in\mathcal{B}\left(\mathcal{H}\right)\,|\,\left[B,U\left(g\right)\right]=0,\,\,\forall g\in G\right\} .
\]
It is an immediate consequence of Schur's lemmas, and the group action
(\ref{eq:U(g)=00003D bigOsum( U_q otimes I)}), that $U\left(G\right)'$
consists of all operators of the form
\begin{equation}
B=\bigoplus_{q}I_{\mathcal{M}_{q}}\otimes B_{\mathcal{N}_{q}}.\label{eq:commutant element of U(G)}
\end{equation}
Compare it to Eq. (\ref{eq:S_q,kl  for symmetries}), from which follows
$U\left(G\right)'=\mathsf{span}\left\{ S_{q,kl}\right\} $. Since
the loss of information under QCG is captured by orthogonal projection
on $\mathsf{span}\left\{ S_{q,kl}\right\} $, it then follows that
the information that is eliminated in this case resides in the degrees
of freedom that are not invariant under the action of the group. 

So far we have established that unitary representations of groups
can be used to specify a QCG scheme. The question remaining is which
groups are useful for the reduction of dynamics. Historically, the
groups that are considered in the study of dynamical processes are
the ones that commute with the dynamics. In the case of unitary time
evolutions, these are the groups that commute with the Hamiltonian
$\left[U\left(g\right),H\right]=0$. In this case $H\in U\left(G\right)'$
so it can be expressed in the form (\ref{eq:commutant element of U(G)})
\begin{equation}
H=\bigoplus_{q}I_{\mathcal{M}_{q}}\otimes H_{\mathcal{N}_{q}}.\label{eq:symmetric H decomposition}
\end{equation}
Dynamics generated by such Hamiltonians keeps the irrep spaces $\mathcal{M}_{q}$
stationary, while evolving the multiplicity spaces $\mathcal{N}_{q}$
independently in each sector. Therefore, the degrees of freedom associated
with the irrep spaces $\mathcal{M}_{q}$ can be safely ignored when
considering time evolutions. From this we conclude that QCG by symmetrization
with the symmetry group of the Hamiltonian is compatible with dynamics.
(This can be shown rigorously by invoking the compatibility condition
of Theorem \ref{thm:=00005BH,S=00005D in span=00007BS_kl=00007D}
and using the fact that $H\in U\left(G\right)'=\mathsf{span}\left\{ S_{q,kl}\right\} $). 

This however, does not mean that symmetries of the Hamiltonian are
the only groups that are useful for the reduction of dynamics. The
appropriate generalization of symmetries of the Hamiltonian, capturing
all groups that can be used to reduce the dynamics, is given in the
following theorem.
\begin{thm}
\label{thm:comp. of symmetrization with dynamocs}Let $G$ be a finite
or a compact Lie group with unitary representation $U\left(G\right)$
on $\mathcal{H}$. Then, coarse-graining by symmetrization with $U\left(G\right)$
is compatible with dynamics generated by the Hamiltonian $H$ if and
only if 
\begin{equation}
\left[U\left(g\right),H\right]\in U\left(G\right)''\,\,\,\,\,\,\,\,\,\,\,\,\,\,\forall g\in G,\label{eq: group action compatibility condition}
\end{equation}
where $U\left(G\right)''$ is the commutant of $U\left(G\right)'$. 
\end{thm}
\begin{proof}
Using the fact that the bipartition operators $\left\{ S_{q,kl}\right\} $
of QCG by symmetrization span $U\left(G\right)'$, we can express
the compatibility condition of Theorem \ref{thm:=00005BH,S=00005D in span=00007BS_kl=00007D}
as

\[
\left[H,B\right]\in U\left(G\right)'\,\,\,\,\,\,\,\,\,\,\,\,\,\,\,\,\,\,\,\,\,\,\,\forall B\in U\left(G\right)'.
\]
By definition of $U\left(G\right)'$, this is equivalent to 
\[
\left[U\left(g\right),\left[H,B\right]\right]=0\,\,\,\,\,\,\,\,\,\,\,\,\,\,\,\,\,\,\,\,\,\,\,\forall B\in U\left(G\right)',\,\forall g\in G.
\]
Since $\left[U\left(g\right),B\right]=0$, we can rearrange the Lie
bracket 
\[
\left[\left[U\left(g\right),H\right],B\right]=0\,\,\,\,\,\,\,\,\,\,\,\,\,\,\,\,\,\,\,\,\,\,\,\forall B\in U\left(G\right)',\,\forall g\in G.
\]
But this means that for all $g$, $\left[U\left(g\right),H\right]$
must be in the commutant of $U\left(G\right)'$, so
\[
\left[U\left(g\right),H\right]\in U\left(G\right)''\,\,\,\,\,\,\,\,\,\,\,\,\,\,\forall g\in G.
\]
Since it is equivalent to the condition of Theorem \ref{thm:=00005BH,S=00005D in span=00007BS_kl=00007D},
which is necessary and sufficient, it is also necessary and sufficient.
\end{proof}
The commutant $U\left(G\right)''$ of the algebra $U\left(G\right)'$
consists of all operators of the form \cite{Knill00}
\[
A=\bigoplus_{q}A_{\mathcal{M}_{q}}\otimes I_{\mathcal{N}_{q}}.
\]
Since all $U\left(g\right)$ are of this form, that is $U\left(g\right)\in U\left(G\right)''$,
condition (\ref{eq: group action compatibility condition}) implies
that groups such that $H\in U\left(G\right)''$ are compatible. Symmetries
of the Hamiltonian, for which $H\in U\left(G\right)'$, trivially
comply with the condition (\ref{eq: group action compatibility condition})
because $0\in U\left(G\right)''$. In general, the compatibility condition
(\ref{eq: group action compatibility condition}) implies a very specific
form for the Hamiltonian.
\begin{prop}
\label{prop:form of compatible Hamiltonian}Any operator $H$ that
complies with the condition (\ref{eq: group action compatibility condition}),
is of the form
\begin{equation}
H=A+B=\bigoplus_{q}\left(A_{\mathcal{M}_{q}}\otimes I_{\mathcal{N}_{q}}+I_{\mathcal{M}_{q}}\otimes B_{\mathcal{N}_{q}}\right),\label{eq:compatible H decomposition}
\end{equation}
where $A\in U\left(G\right)''$ and $B\in U\left(G\right)'$. 
\end{prop}
\begin{proof}
Condition (\ref{eq: group action compatibility condition}) implies
that for every $g\in G$ there is an $A_{g}\in U\left(G\right)''$
s.t.
\[
U\left(g\right)HU\left(g\right)^{\dagger}-H=A_{g}.
\]
Rearranging the terms and integrating over $G$ (summing for finite
groups) with an invariant measure $d\mu\left(g\right)$ we get
\[
H=\underset{A}{\underbrace{-\int_{G}d\mu\left(g\right)A_{g}}}+\underset{B}{\underbrace{\int_{G}d\mu\left(g\right)U\left(g\right)HU\left(g\right)^{\dagger}}}.
\]
We have $A\in U\left(G\right)''$ by definition of $A_{g}$, and $B\in U\left(G\right)'$
because of the invariance of the measure $d\mu\left(g\right)=d\mu\left(g'\right)$
(or Rearrangement Theorem for finite groups). 
\end{proof}
It is now easy to see why groups that comply with condition (\ref{eq: group action compatibility condition})
lead to compatible QCG by symmetrization. The form (\ref{eq:compatible H decomposition})
implies that the subsystems $\mathcal{M}_{q}$ and $\mathcal{N}_{q}$
do not interact with each other. The explicit form of the time evolution
operator is 

\[
U_{H}\left(t\right)=e^{-itH}=\bigoplus_{q}e^{-itA_{\mathcal{M}_{q}}}\otimes e^{-itB_{\mathcal{N}_{q}}},
\]
so each part of the virtual composite system $\mathcal{M}_{q}\otimes\mathcal{N}_{q}$
evolves independently from the other. Therefore, we can generate time
evolutions in $\mathcal{N}_{q}$ without having to know the state
of $\mathcal{M}_{q}$ (and vise versa), and that is the defining property
needed for the reduction of dynamics. 

Symmetries of the Hamiltonian ($\left[U\left(g\right),H\right]=0$)
are too restrictive for the purposes of reduction of dynamics. They
require that, in addition to subsystem $\mathcal{N}_{q}$ evolving
independently, subsystem $\mathcal{M}_{q}$ must be stationary, which
is not necessary. Relaxing the condition to (\ref{eq: group action compatibility condition}),
and letting $\mathcal{M}_{q}$ evolve, leads to a broader set of groups,
beyond symmetries of the Hamiltonian. Thus, Theorem \ref{thm:comp. of symmetrization with dynamocs}
provides us with more possibilities to confine the dynamical evolutions
to smaller state-spaces. 

For practical applications it is beneficial to express the compatibility
condition (\ref{eq: group action compatibility condition}) in terms
of the generators of the group. Assuming $\left\{ L_{\alpha}\right\} $
are the generators of $U\left(G\right)$, and using the group action
near the identity $U\left(\epsilon_{\alpha}\right)=I-\epsilon_{\alpha}iL_{\alpha}$
(for finite groups we can use the generators directly), the compatibility
condition becomes
\[
\left[L_{\alpha},H\right]\in U\left(G\right)''\,\,\,\,\,\,\,\,\,\,\,\,\,\,\,\,\,\forall L_{\alpha}.
\]
Furthermore, the operator algebra $\mathsf{Alg}\left\{ L_{\alpha}\right\} $,
of all polynomials in $\left\{ L_{\alpha}\right\} $, is a subalgebra
of $U\left(G\right)''$ (by definition of $U\left(G\right)'$, every
$L_{\alpha}$ must commute with everything in $U\left(G\right)'$).
Thus, replacing $U\left(G\right)''$ with $\mathsf{Alg}\left\{ L_{\alpha}\right\} $,
results in the sufficient condition
\[
\left[L_{\alpha},H\right]\in\mathsf{Alg}\left\{ L_{\beta}\right\} \,\,\,\,\,\,\,\,\,\,\,\,\,\,\,\,\,\forall L_{\alpha}.
\]

\subsection{Example: Continuous-time Quantum Walk on a Binary Tree}

Continuous-time quantum walk (CTQW) is a generic model of quantum
dynamics that admits visually intuitive demonstration of QCG by symmetrization.
More specifically, we will focus on CTQW on binary trees introduced
by \cite{Farhi98} and demonstrated to evolve in a reduced state-space
in \cite{Childs01}. 

The CTQW model is specified by a simple undirected graph $\mathcal{G}$
with vertices $V$ and edges $E$. The Hilbert space is defined as
$\mathcal{H}:=\textsf{span}\left\{ \ket{v_{i}}\right\} _{v_{i}\in V}$,
and the Hamiltonian is constructed in the same way as the stochastic
transition rates matrix (in this case all rates are normalized to
$1$) 
\begin{equation}
H:=-\sum_{\left(v_{i},v_{j}\right)\in E}\left(\ket{v_{i}}\bra{v_{j}}+\ket{v_{j}}\bra{v_{i}}\right)+\sum_{v_{i}\in V}d_{i}\ket{v_{i}}\bra{v_{i}}.\label{eq:CTQW hamiltonian def}
\end{equation}
The degree $d_{i}$ of a vertex $v_{i}$ is the total number of vertices
connected to it. 

The concrete example we will analyze is shown in fig. \ref{fig:BinaryTreeWalk}(a),
\begin{figure}[t]
\begin{raggedright}
(a) %
\noindent\begin{minipage}[t]{1\columnwidth}%
\includegraphics[width=1\columnwidth]{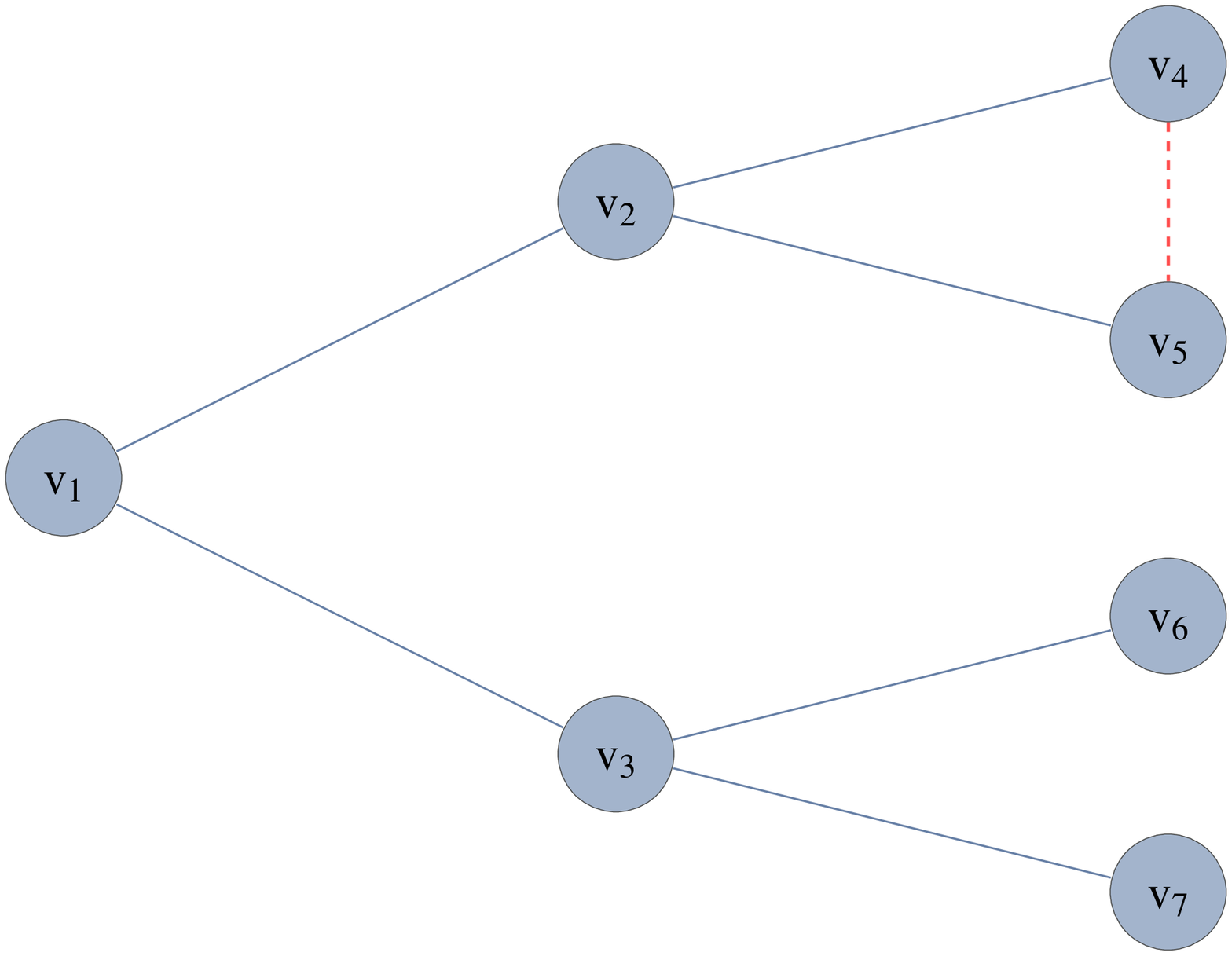}%
\end{minipage}
\par\end{raggedright}
\raggedright{}(b) %
\noindent\begin{minipage}[t]{1\columnwidth}%
\includegraphics[width=1\columnwidth]{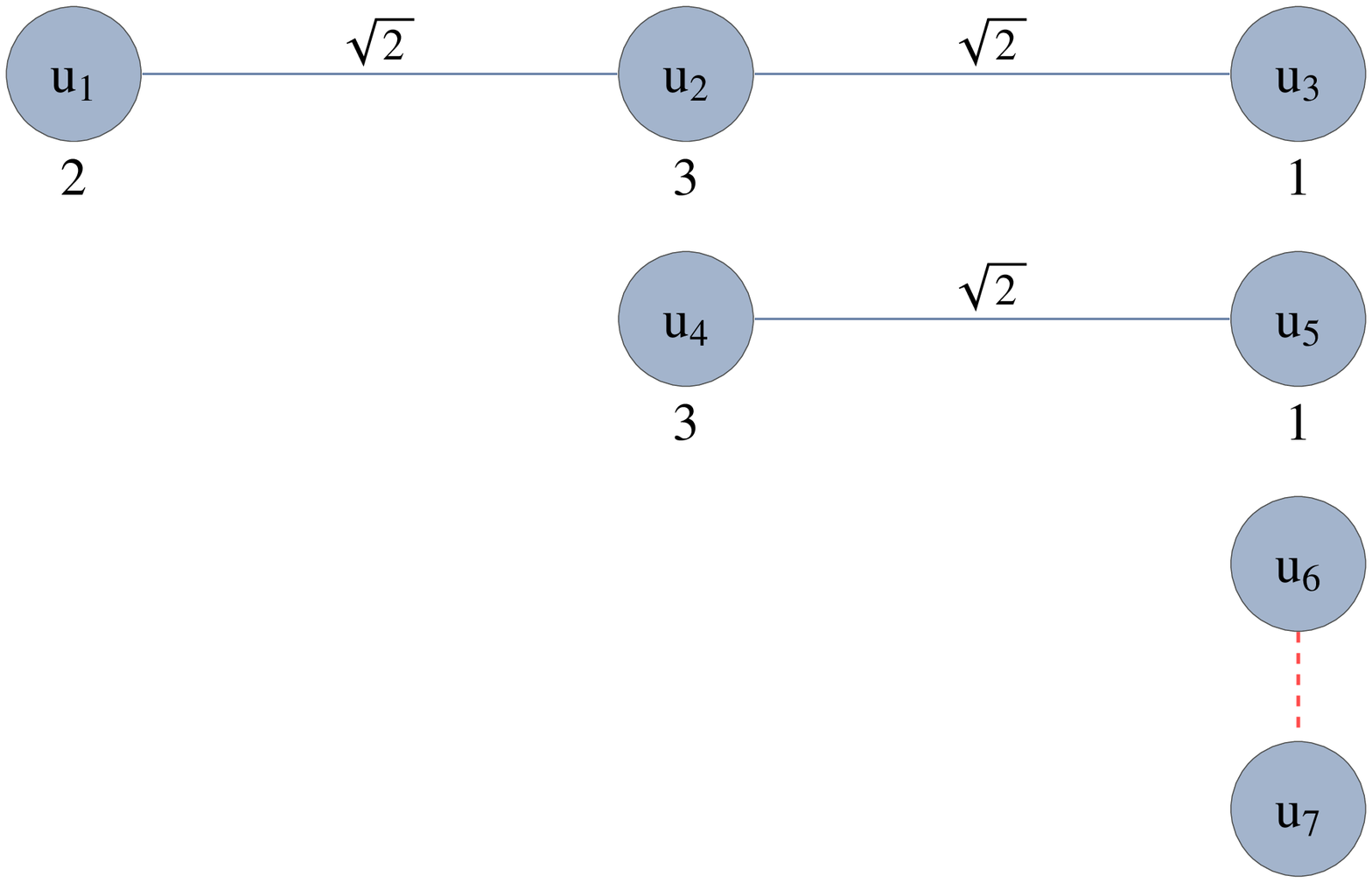}%
\end{minipage}\caption{\label{fig:BinaryTreeWalk}Quantum walk on a tree (a) is reduced to
parallel quantum walks over the columns of the tree (b). Labels on
the edges indicate transition rates and labels under the vertices
indicate local potentials (the default value is 1 for both). Addition
of the red dashed edge in (a) results in the addition of identical
edge in (b). Even though the red edge brakes the symmetry of the tree
it does not affect the reduced dynamics between the columns (see main
text).}
\end{figure}
where we ignore the red dashed line for now and focus on the tree
spanned by the solid edges. If this was a classical random walk, then
CG by partition of vertices to the 3 columns would be compatible with
dynamics. In the quantum case, however, partition to sectors is not
enough to specify a CG, and there is not much else to guide us in
the appropriate choice of compatible bipartition other than symmetries. 

Symmetries of CTQWs arise from the automorphisms of the underlying
graph \cite{Krovi07}. Graph automorphisms form a group $\mathsf{Aut}_{\mathcal{G}}:=\left\{ \varphi\right\} $
that consists of permutations of vertices that leave the set of edges
unchanged 
\[
\left(v_{i},v_{j}\right)\in E\,\,\,\,\,\Leftrightarrow\,\,\,\,\,\left(v_{\varphi\left(i\right)},v_{\varphi\left(j\right)}\right)\in E.
\]
Using the cycle notation for permutations, our graph automorphisms
are generated by $a=\left(45\right)$ and $b=\left(23\right)\left(47\right)\left(56\right)$;
it is also instructive to point out the group element $c=bab=\left(67\right)$.
Permutation $b$ can be thought of as a flip of the whole tree around
the horizontal axes through the root, while permutations $a$ and
$c$ are flips of the sub-trees around horizontal axes through their
own roots. 

In order to streamline the calculations, it is convenient to express
the Hamiltonian as a sum of permutations. Permutations are naturally
represented by orthogonal (unitary) operators 
\[
\Pi_{\varphi}:=\sum_{i}\ketbra{v_{\varphi\left(i\right)}}{v_{i}}.
\]
The Hamiltonian (\ref{eq:CTQW hamiltonian def}) can now be written
as a sum of 2-cycle permutations $\left(ij\right)$ (note that $\Pi_{\left(ij\right)}$
acts as the identity on vertices that are not $v_{i}$ or $v_{j}$)
\[
H=-\sum_{\left(v_{i},v_{j}\right)\in E}\Pi_{\left(ij\right)}+\left|E\right|I.
\]
Since $\left|E\right|I$ only adds a total phase to the evolutions
we can safely drop it. In our concrete case the Hamiltonian is 
\begin{equation}
H=-\Pi_{\left(12\right)}-\Pi_{\left(13\right)}-\Pi_{\left(24\right)}-\Pi_{\left(25\right)}-\Pi_{\left(36\right)}-\Pi_{\left(37\right)}.\label{eq:CTQW hamiltonian in permutations}
\end{equation}
Note that the adjoint action of any permutation $\varphi$ on a 2-cycle,
results in another 2-cycle
\[
\Pi_{\varphi}\Pi_{\left(ij\right)}\Pi_{\varphi}^{T}=\Pi_{\left(\varphi\left(i\right)\varphi\left(j\right)\right)}.
\]
It is now easy to check that the group generated by $a=\left(45\right)$
and $b=\left(23\right)\left(47\right)\left(56\right)$ commutes with
the Hamiltonian, because the adjoint action of $\Pi_{a}$ or $\Pi_{b}$
permutes the 2-cycles in (\ref{eq:CTQW hamiltonian in permutations}),
but leaves the whole sum unchanged
\[
\Pi_{a}H\Pi_{a}^{T}=\Pi_{b}H\Pi_{b}^{T}=H.
\]
Therefore, the finite group $\mathsf{Aut}_{\mathcal{G}}$ represented
by the unitary operators $\left\{ \Pi_{a},\Pi_{b}\right\} $ is a
symmetry of the Hamiltonian. 

Using the shorthand notation 
\[
\ket{+_{ijk...}}=\frac{\ket{v_{i}}+\ket{v_{j}}+\ket{v_{k}}+...}{normalization}
\]
\[
\ket{\pm_{ijk...}}=\frac{\ket{v_{i}}-\ket{v_{j}}+\ket{v_{k}}-...}{normalization}
\]
we first identify the 3 trivial irreps of $\mathsf{Aut}_{\mathcal{G}}$
as the subspaces
\begin{align*}
\mathcal{M}_{1,1} & :=\textsf{span}\left\{ \ket{+_{1}}\right\} \\
\mathcal{M}_{1,2} & :=\textsf{span}\left\{ \ket{+_{23}}\right\} \\
\mathcal{M}_{1,3} & :=\textsf{span}\left\{ \ket{+_{4567}}\right\} .
\end{align*}
There are also 2 non-trivial but equivalent irreps, where $\Pi_{a}$
acts by $1$ and $\Pi_{b}$ acts by $-1$ 
\begin{align*}
\mathcal{M}_{2,1} & :=\textsf{span}\left\{ \ket{\pm_{23}}\right\} \\
\mathcal{M}_{2,2} & :=\textsf{span}\left\{ \ket{\pm_{4657}}\right\} .
\end{align*}
The last irrep is single and 2-dimensional 
\begin{align*}
\mathcal{M}_{3} & :=\textsf{span}\left\{ \ket{\pm_{4567}},\ket{\pm_{5467}}\right\} .
\end{align*}

Accounting for multiplicities, the Hilbert space decomposes to
\[
\mathcal{H}=\left(\mathcal{M}_{1}\otimes\mathcal{N}_{1}\right)\oplus\left(\mathcal{M}_{2}\otimes\mathcal{N}_{2}\right)\oplus\mathcal{M}_{3}
\]
where $\mathcal{N}_{1}$ and $\mathcal{N}_{2}$ are $3$ and $2$
dimensional multiplicity spaces. Now we can change to the new basis
$\ket{u_{i}}$ that are native to these irreps
\[
\begin{array}{lclcc}
\ket{u_{1}}:=\ket{+_{1}} &  & \ket{u_{4}}:=\ket{\pm_{23}} &  & \ket{u_{6}}:=\ket{\pm_{4567}}\\
\ket{u_{2}}:=\ket{+_{23}} &  & \ket{u_{5}}:=\ket{\pm_{4657}} &  & \ket{u_{7}}:=\ket{\pm_{5467}}\\
\ket{u_{3}}:=\ket{+_{4567}}.
\end{array}
\]
In the new basis the Hamiltonian is block diagonal $H=H_{1}\oplus H_{2}\oplus H_{3}$,
where
\begin{equation}
H_{1}=\begin{pmatrix}2 & -\sqrt{2} & 0\\
-\sqrt{2} & 3 & -\sqrt{2}\\
0 & -\sqrt{2} & 1
\end{pmatrix}\,\,\,\,\,\,\,\,\,\,\begin{array}{l}
H_{2}=\begin{pmatrix}3 & -\sqrt{2}\\
-\sqrt{2} & 1
\end{pmatrix}\\
\\
H_{3}=\begin{pmatrix}1 & 0\\
0 & 1
\end{pmatrix}.
\end{array}\label{eq:Tree walk Hamiltonian blocks}
\end{equation}
(We added back the global term $\left|E\right|I$ to present the more
conventional diagonal elements). This is the explicit form (\ref{eq:symmetric H decomposition})
of $H$ that acts non-trivially on multiplicity spaces only. The terms
$H_{1}$, $H_{2}$ act on multiplicity spaces of 1-dimensional irreps,
and the term $H_{3}$ acts trivially because $\mathcal{M}_{3}$ is
multiplicity free. Therefore, the dynamics can be isolated as quantum
walks on disconnected components associated with the irreps, see fig.
\ref{fig:BinaryTreeWalk}(b). One complication that arises, caused
by boundaries of the finite graph, is the non-constant potential on
the vertices, as seen on the diagonal of the Hamiltonian. 

Quantum walk in the multiplicity space of the trivial irrep (top row
in fig \ref{fig:BinaryTreeWalk}(b)) was first shown in \cite{Childs01}
to be the reduced 1D walk over the ``column states'' ($\ket{u_{1}}$,
$\ket{u_{2}}$, $\ket{u_{3}}$ in our notation). Boundary effects,
causing a potential ``bump'', were numerically shown to be not significant
in the larger trees. More importantly, it was understood that the
reduced quantum walk on the 1D chain of ``column states'' is responsible
for the exponential speedup in propagation time from the leafs to
the root, compared to the classical walk. Note however, that the full
speedup occurs only from the initial state $\ket{u_{3}}$. If the
initial state is also supported on $\ket{u_{5}}$ then the speedup
will only carry it to the second column, and if it had support on
$\ket{u_{6}}$ or $\ket{u_{7}}$ then those parts are stuck in the
initial column. 

The new insight is that the reduced dynamics of quantum walks on trees
persist even if the symmetry of the tree is broken in a manner that
is described by Eq. (\ref{eq: group action compatibility condition}).
Adding the red dashed edge to the tree in fig \ref{fig:BinaryTreeWalk}(a)
breaks the original automorphism symmetry since $\left(v_{4},v_{5}\right)\in E$
but $\left(v_{b\left(4\right)},v_{b(5)}\right)=\left(v_{7},v_{6}\right)\notin E$.
The new Hamiltonian $H'$, when expressed as a sum of 2-cycles, receives
an additional term $H'=H-\Pi_{\left(45\right)}$, which breaks the
symmetry under the action of $b$ 
\begin{align*}
\Pi_{b}H'\Pi_{b}^{T} & =H-\Pi_{\left(76\right)}\neq H'.
\end{align*}
The action of $a=\left(45\right)$, however, still commutes with $H'$.
The commutator of $H'$ with $\Pi_{b}$ can be expressed as 
\[
\left[\Pi_{b},H'\right]=\left[\Pi_{b},H\right]-\left[\Pi_{b},\Pi_{a}\right]=\left[\Pi_{a},\Pi_{b}\right],
\]
so it belongs to the operator algebra spanned by the generators $\{\Pi_{a},\Pi_{b}\}$.
Theorem \ref{thm:comp. of symmetrization with dynamocs} then implies
that QCG by such symmetrization is still compatible with dynamics.
The only difference is that in addition to acting on the multiplicity
spaces, the Hamiltonian may act independently on the irreps. In this
case, only $\mathcal{M}_{3}$ can be affected (dynamics in 1 dimensional
irreps are absorbed into the multiplicity spaces). Even though $\ket{u_{i}}$
are now native to irreps of a group that is not a symmetry of the
Hamiltonian, we can still use them to block diagonalize the Hamiltonian.
The new decomposition is $H'=H_{1}\oplus H_{2}\oplus H_{3}'$ where
$H_{1}$, $H_{2}$ are the same as before (\ref{eq:Tree walk Hamiltonian blocks}),
and
\[
H_{3}'=\begin{pmatrix}2 & -1\\
-1 & 2
\end{pmatrix}.
\]
$H_{3}'$ generates non-trivial evolutions in $\mathcal{M}_{3}$ which
can be seen graphically as addition of the red dashed edge in fig.
\ref{fig:BinaryTreeWalk}(b). It generates evolutions vertically,
in a stationary subspace of the right column, but it does not interfere
with dynamics across the columns.

This example demonstrates the fact that strict symmetries of the Hamiltonian
are not necessary for the effective reduction of dynamics. That being
said, it is not easy to see a priori which groups are compatible.
This wouldn't work, for example, if we broke the symmetry with $\Pi_{(56)}$
instead of $\Pi_{\left(45\right)}$, since $\left[\Pi_{\left(56\right)},\Pi_{a}\right]$
is not an element in the operator algebra spanned by $\{\Pi_{a},\Pi_{b}\}$.
Just because $\Pi_{(56)}$ generates dynamics within the column does
not mean that it cannot interfere with the dynamics across the columns.
In our case, the choice of $\Pi_{\left(45\right)}$ to brake the symmetry
works, because it is the element $a$ of the symmetry group of the
Hamiltonian. Since $\mathsf{Aut}_{\mathcal{G}}$ is not abelian, the
modified Hamiltonian was no longer commuting with it, but because
the symmetry breaking element came from the group, the commutant was
guaranteed to be in the operator algebra spanned by the group.

\section{Summary and Outlook}

We have established the common notion of coarse-graining in both classical
and quantum settings and provided it with operational meaning. By
introducing bipartition tables we were able to capture the key structure
of a coarse-graining scheme in a concise, visual form. Our main focus
\textendash{} the reduction of dynamics by coarse-graining the state-space
\textendash{} lead to the formulation of compatibility conditions
between a coarse-graining scheme and dynamics. Such compatibility
conditions were shown to be necessary and sufficient for the existence
of a reduced generator of dynamics that governs time evolutions in
the coarse-grained state-space. Considering symmetrizations of states
with a group representation as a special case of coarse-graining,
and specializing the compatibility condition to this case, we showed
how group representations can be used to reduce the dynamics. This
result turned out to be closely related to Noether's Theorem that
uses symmetries of dynamics to identify the static degrees of freedom
i.e., constants of motion. We generalized this perspective with less
restricted group representations to identify dynamically independent
degrees of freedom. Such degrees of freedom are not necessarily constants
of motion and the group representations are not necessarily symmetries
of dynamics.

The task of reducing dynamics that was studied here demands an exact
reproduction of dynamical evolutions in the reduced state-space. The
only way to satisfy such demand is to single out the degrees of freedom
that evolve independently from the rest. As we pointed out, finding
group representations that satisfy the commutation relation (\ref{eq: group action compatibility condition})
is one possible approach to the problem of exact reduction. This formulation,
however, might be too strict for some practical application and an
approximate reduction may be in order. The compatibility conditions
for the exact reduction can then be taken as a starting point for
the development of approximations when the conditions are not exactly
satisfied. 

Aside from the reduction of dynamics, the notion of coarse-graining
raises some interesting questions on its own. QCG was shown to be
the map that accounts for some ignorance of the observer. Specifically,
it accounts for the restriction to measurements that belong to the
span of bipartition operators that define the QCG scheme. Such restricted
observers arise naturally in physical situations characterized by
inability to measure external environment, or inaccessibility of a
particular reference frame in which the system is prepared. The QCG
formalism allows to account for all these and more general restrictions,
but the physical situations that lead to the more general cases are
not so clear. In particular, we can now account for restrictions to
observables that do not form an algebra \textendash{} the general
span of bipartition operators is an \textit{operator system.} Then
the question is, what physical situations lead to the restriction
to observables that form an operator system, but not an algebra?

Regardless of physical interpretations, QCG is a powerful analytical
concept that offers new and flexible ways to ``select'' quantum
information encoded in the physical state. Within the framework of
QCG we can capture information selected by group representations,
virtual subsystems, and restricted observables under the same umbrella
concept. We believe that this generic nature makes QCG a fundamental
concept of quantum information with potential applications in quantum
error correction, tomography and quantum thermodynamics. 
\begin{acknowledgments}
I would like to thank R. Raussendorf for guidance and review throughout
preparation of this work. This work was supported by NSERC.
\end{acknowledgments}

\end{document}

%% file: MathMacros.tex
\global\long\def\ket#1{\left|#1\right\rangle }

\global\long\def\bra#1{\left\langle #1\right|}

\global\long\def\braket#1#2{\left.\left\langle #1\right.\,\right|\left.#2\right\rangle }

\global\long\def\ketbra#1#2{\left|#1\right\rangle \left\langle #2\right|}